\documentclass{article}
\usepackage{amsmath, amsfonts, amsthm, geometry, tikz, xcolor, graphicx, hyperref, cleveref, thmtools, thm-restate}
\usepackage{mathrsfs}

\theoremstyle{plain}
\theoremstyle{definition}
\newtheorem{protocol}{Protocol}
\newtheorem{theorem}{Theorem}[section]
\newtheorem{lemma}[theorem]{Lemma}
\newtheorem{proposition}[theorem]{Proposition}

\theoremstyle{definition}
\newtheorem{definition}[theorem]{Definition}

\theoremstyle{remark}
\newtheorem{remark}[theorem]{Remark}

\newcommand{\poly}{\mathrm{poly}}

\newcommand{\F}{\mathbb F}

\newcommand{\N}{\mathbb N}

\newcommand{\cD}{\mathcal D}
\newcommand{\cC}{\mathcal C}
\newcommand{\cH}{\mathcal H}
\newcommand{\dist}{\mathsf{dist}}

\newcommand{\eps}{\varepsilon}
\renewcommand{\epsilon}{\varepsilon}

\renewcommand{\Pr}{\mathbb{P}}
\newcommand{\Prop}[1]{\Pr\left[#1\right]}
\newcommand{\wt}{\mathrm{wt}}

\title{Fast Bounded-Independence Functions and Their Duals\thanks{This is a full version of~\cite{ITCversion}.}}
\author{Martijn Brehm\thanks{Informatics Institute, University of Amsterdam. \href{m.a.brehm@uva.nl}{m.a.brehm@uva.nl}.} \and Yuval Ishai\thanks{Technion and AWS. \href{yuvali@cs.technion.ac.il}{yuvali@cs.technion.ac.il}.  This work is not associated with Amazon. Research supported by ISF grant 3527/24 and BSF grant 2022370.  } \and Nicolas Resch\thanks{Informatics Institute, University of Amsterdam. \href{n.a.resch@uva.nl}{n.a.resch@uva.nl}. Research supported by an NWO (Dutch Research Council) grant with number C.2324.0590.}}

\begin{document}

\maketitle

\begin{abstract}
We continue the study of {\em fast} functions, computable by linear-size circuits, that share useful properties of random functions. Motivated by cryptographic applications, we generalize and improve on previous results in this area, obtaining the following results:
\begin{itemize}
\item For any constant $t$, we construct a fast $t$-wise independent hash function with algebraic degree $\log_2 t$ (over $\F_2$), simultaneously optimizing both asymptotic circuit size and degree.
\item We simplify and improve a recent construction (ITCS 2026) of a family of fast codes with fast duals, both meeting the Gilbert-Varshamov bound. Unlike the previous construction, our construction has negligible failure probability, can accommodate general fields and rates, supports a systematic encoding, and admits fast universal encoders.
\item We strengthen the above to support stronger random-like properties, such as optimal combinatorial list-decoding. This is achieved by constructing, for any constant $t$, a family of fast linear functions that map any $t$ linearly independent inputs to uniform and statistically independent outputs. Prior to our work, this was only known for $t=1$.

\end{itemize}

We demonstrate the usefulness of the above results to cryptography. This includes the first nontrivial protocols for perfectly secure multiparty computation whose circuit complexity scales linearly with the number of parties, as well as protocols for computing encrypted matrix-vector products with optimal asymptotic circuit complexity.

\end{abstract}

\newpage

\tableofcontents

\newpage

\section{Introduction}

An important trend in cryptographic research is to develop cryptographic schemes realizing functionalities as efficiently as possible. Ideally, the cost of performing a task \emph{with} security should be on the order of the cost of performing the same task \emph{without} security. That is, security just results in a \emph{constant} multiplicative computational overhead. Originating from~\cite{IKOS08}, this has been the theme of a line of work on cryptography with constant computational overhead; see~\cite{BoyleCGIKRS23} and the references therein.

As a concrete example, consider the task of hashing a large amount of data, where one generates a short digest from a long string.
One typically wants the hash function to satisfy certain desirable properties of a random function. Ideally, the cost of performing the hashing should only scale linearly with the size of the data. Which ``random-like'' properties can be satisfied while achieving constant computational overhead compared to just reading the input?

Similarly, one can consider the task of {\em encoding} a large vector in a way that satisfies desirable properties of a random linear code. In this case, one is sometimes also required to apply the encoding function of the {\em dual} code. Is it possible to implement {\em both} the primal and the dual encoding with constant overhead? 

The above hashing and encoding primitives serve as useful building blocks in many algorithmic and cryptographic applications. Minimizing their overhead is relevant to all of these applications. 

\subsection{Our Results} \label{subsec:our-results}

To make progress on the above questions, we provide new constructions that achieve {\em constant computational overhead} for random-like functions. This refers by default to implementations by (fan-in 2) Boolean or arithmetic circuits whose size scales linearly with the input and output length; see~\Cref{subsec:computational} for discussion. We refer to such an implementation as being {\em fast}.

For hash functions, the work of Ishai, Kushilevitz, Ostrovsky, and Sahai~\cite{IKOS08} demonstrated the existence of fast {\em pairwise independent} hash functions with arbitrary input and output length.

We first demonstrate that one can build on this result to instead obtain fast \emph{$t$-wise independent} hash functions. That is, even after seeing $t-1$ hash values, the $t$-th value is completely unpredictable (assuming the input is distinct from the previous $t-1$ inputs). This result has already been discussed informally in talks by authors of~\cite{IKOS08}. Here we spell it out in detail, also achieving near-optimal algebraic degree as a function of $t$, which is useful for cryptographic applications. This result also acts as a nice warmup for our later results. 

\begin{theorem} [{Fast $t$-wise hash function; informal version of \Cref{thm:t-wise-hash}}]
	For any constant $t\ge 2$, finite field $\F_q$, input length $k$, and output length $r$, there exists a construction of a fast $t$-wise independent hash function family $\cH$ of functions $h:\F_q^k \to \F_q^r$. For $q=2$, this can be done with (optimal) algebraic degree of $\log_2 t$. 
\end{theorem}

\paragraph{Fast dual codes.} Partially motivated by 
cryptographic applications, a recent work of Brehm and Resch~\cite{BR26} constructed a family of linear codes $\cC\leq \F_2^n$ of rate $1/2$ with the following properties:
\begin{itemize}
	\item Both the code and its dual admit \emph{fast} encoding;
	\item With probability $\geq 1 - 1/\poly(n)$, the code \emph{and its dual} $\cC^\perp$ achieve distance $\eps$-close to the Gilbert-Varshamov (GV) bound. 
\end{itemize}

Recall that a code $\cC$ of rate $R$ achieves distance $\eps$-close to the GV bound, if its (relative) distance is at least $h_q^{-1}(1-R-\eps)$, where $h_q$ is the $q$-ary entropy function (and $h_q^{-1}$ is its inverse); see \Cref{subsec:coding-theory} for more discussion. The GV bound captures the distance achieved by \emph{uniformly} random linear codes, which is typically the best distance one could reasonably hope for. 

The \cite{BR26}~result was based on a rather involved analysis of a code constructed in a manner inspired by repeat-multiple-accumulate codes~\cite{DJM98}. While it provided good concrete efficiency that can be useful for applications (see~\Cref{rem:efficiency}), this result also has several drawbacks:
\begin{itemize}
	\item The code is only constructed over the binary field $\F_2$;
	\item The code (and its dual) need to have rate exactly $1/2$;
    \item The randomized construction of the code has a non-negligible failure probability;
	\item The encoding map is not systematic (which is a problem for some applications);
    \item The construction does not provide a fast {\em universal circuit} for encoding given both the message and the randomness defining the code.
\end{itemize}

Building on the work of Druk and Ishai~\cite{DI14} (which, in turn, builds on~\cite{IKOS08}), we provide a simpler alternative construction that does not suffer from any of these drawbacks. 

\begin{theorem} [Fast code with fast dual; informal version of \Cref{thm:first_dual}]
	For any $R \in (0,1)$ and any prime power $q$, there is a family of codes $\cC\leq\F_q^n$ of rate $R$ along with their duals $\cC^\perp$, for which:
	\begin{itemize}
		\item Both $\cC$ and $\cC^\perp$ admit fast systematic encoding;
		\item $\forall \eps>0$, $\cC$ and $\cC^\perp$ are both $\eps$-close to the GV bound with probability $\geq 1-2q^{-\eps n}$. 
	\end{itemize}
    Furthermore, the encoding matrices for $\cC$ and $\cC^\perp$ are determined by a linear-sized universal circuit. 
\end{theorem}

\begin{remark}
\label{rem:efficiency}
	A significant advantage of the construction from~\cite{BR26} is concrete efficiency: encoding can be done by circuits of size $12n$. In contrast, the current approach involves much bigger hidden constants that we did not attempt to optimize.
	
	As an additional comment, for both constructions the encoding circuit has depth $O(\log n)$, which is asymptotically optimal in the bounded fan-in case. However, here too we expect our construction to have a much bigger hidden constant than the previous construction from~\cite{BR26}. 
\end{remark}

\paragraph{Fast, random-like codes.} 
By building off the concept of \emph{local similarity}~\cite{MRRSW24,GM22,MRSY25} -- which informally allows one to argue that A code ensemble ``behaves like'' a uniformly random linear code -- we can show that we can extend upon this construction, to obtain $\cC$ and $\cC^\perp$ possessing stronger properties. For example, they can be proved to have very good \emph{list-decodability}, up to capacity, with optimal list sizes. Recall that a code $\cC \leq \F_q^n$ is said to be $(\rho,L)$-list-decodable if for all $z \in \F_q^n$ the number of codewords within distance $\rho$ from $z$ is at most $L$. Additionally, the rate $1-h_q(\rho)$ is called the list-decoding capacity: informally, it is the maximum rate at which a code can be $(\rho,L)$-list-decodable with any ``reasonable'' (i.e., sub-exponential) list-size $L$. We say a code is $\eps$-close to list-decoding capacity if it has rate at least $1-h_q(\rho)-\eps$. 

\begin{theorem} [Fast list-decodable code with fast dual; informal version of \Cref{thm:second_dual}]
    For any $R \in (0,1)$ and any prime power $q$, there is a randomized procedure producing a code $\cC \leq \F_q^n$ of rate $R$ for which: 
    \begin{itemize}
        \item Both $\cC$ and $\cC^\perp$ admit fast systematic encoding;
        \item $\forall \eps>0$, both $\cC$ and $\cC^\perp$ are $\eps$-close to capacity, with list-size $O(1/\eps)$ with probability $\geq 1-q^{-\Omega(\eps n)}$. 
    \end{itemize}
    Furthermore, the encoding matrices for $\cC$ and $\cC^\perp$ are determined by a linear-sized universal circuit. 
\end{theorem}

We comment here that list-size $O(1/\eps)$ is the best one can ``reasonably'' hope for, in the sense that there is no known construction of a code $\eps$-close to list-decoding capacity with list-size $o(1/\eps)$. For context, prior work had managed to construct linear-time encodable codes with such list-decodability~\cite{GM22}, and also managed to construct a code where it and its dual achieve such list-decodability~\cite{MRSY25}: getting both properties simultaneously is, to the best of our knowledge, novel. 

\paragraph{Cryptographic applications.} While our results are broadly useful in any contexts in which error-correcting codes or hash functions are used, they are particularly motivated by several cryptographic applications. 

In the context of information-theoretic cryptography, any linear code along with a dual code can support {\em secure multiplication} in secure multiparty computation (MPC)~\cite{massey1995some,CDM00}. Our results imply the first perfectly secure MPC protocols in which both the security threshold and the computational complexity (measured by Boolean circuit size) scale linearly with the number of parties. Concretely, for any constant-size function $f$ whose inputs and outputs are owned by a constant number of parties, we get a perfectly secure $n$-party protocol, with a near-optimal (passive) corruption threshold of $t=(1/2-\eps)\cdot n$, in which the total circuit size of all parties is $O(n)$. Our result on fast and minimal-degree $t$-wise hash functions can also be useful for simultaneously minimizing communication and round complexity in distributed MPC implementations.

In the context of complexity-based cryptography, our work (as well as the related~\cite{BR26}) is motivated by a recent technique for computing {\em encrypted matrix-vector products}~\cite{BenhamoudaCHIKM25}. This technique relies on two dual codes, where one code is used for encrypting the matrix and the other is used for encrypting the vector. Our construction of fast dual codes gets around some of the limitations of the previous construction from~\cite{BR26} which are relevant to this use case.

See~\Cref{sec:apps} for details on these and other cryptographic applications. 

\subsection{Technical Overview} \label{subsec:our-techniques}

In this section, we briefly discuss how we obtain our main results. First, we recall the IKOS construction of linear-sized pairwise independent hash function~\cite{IKOS08}, and then the DI construction of a linear-time encodable code getting $\eps$-close to the GV bound~\cite{DI14}, which itself builds off the IKOS construction. 

Firstly, the IKOS construction essentially goes as follows:
\begin{enumerate}
    \item First, apply a fast code with large distance to the input. The alphabet for this code will be large, but constant. 
    \item Next, apply a (constant-sized) pairwise-independent hash function to each of the coordinates.
    \item Lastly, apply an extractor for bit-fixing sources, which is obtained from the transpose of a generator matrix for a good code. In particular, one can take the transpose of the code applied in the first step (this also guarantees this step is fast, thanks to a ``transposition'' principle -- cf. \Cref{lem:fast_tranpose}).
\end{enumerate}
A visual depiction of this construction can be found in Figure~\ref{fig:ikos_construction}. 

The intuition for the pairwise independence is as follows. Given two distinct inputs, the first encoding step produces codewords that have many distinct coordinates. On these distinct coordinates, the constant-sized hash functions map the coordinates to independent, uniform values. The last step then extracts from these many coordinates with independent values a single string of independent values. 

To get the $t$-wise independent hash function family, we simply have to replace the constant-sized pairwise-independent hash functions by constant-sized $t$-wise independent hash functions. The reasoning now is very similar. Suppose now we are given $t$ distinct inputs. If the initial code distance is good enough (greater than $1-1/t^2$), then in many of the coordinates all the values of the encodings will be distinct, allowing the inner hash functions to output independent, uniform randomness which can be extracted in the last step. 

\medskip

We now turn to the constructions of the codes. Here, we make use of the concept of a \emph{linear uniform output family (LUOF)}, as defined in~\cite{DI14}. This is a random variable $A$ distributed over $k \times r$ matrices such that, for any fixed \emph{nonzero} vector $x \in \F_q^k$, $xA$ is distributed uniformly over $\F_q^r$. We recall that if $A$ is uniformly random, then it has this property. In this work -- as in~\cite{DI14} -- we look for distributions over $k \times r$ matrices such that the map $x \mapsto xA$ is \emph{fast} (informally, implementable by a linear-sized circuit). For brevity, we will call such an LUOF \emph{fast}.

The basic idea of~\cite{DI14} is to construct a fast LUOF over $k \times n$ matrices. It is relatively straightforward that this suffices for obtaining codes $\eps$-close to the GV bound.\footnote{Indeed, one can inspect the ``standard'' argument that random linear codes -- i.e., codes defined by sampling a uniformly random generator matrix -- get $\eps$-close to the GV-bound with high probability, and observe that the only property one uses is that uniformly random matrices give an LUOF.} To get the fast LUOF, \cite{DI14} keep the first and third steps the same as in the IKOS construction, but instead of applying pairwise independent hash functions, they simply apply constant-sized matrices to each of the coordinates, which are viewed as length-$\beta$ vectors over the base field $\F_q$ (recall the first encoding step may increase the alphabet size, but only by a constant amount). The reasoning for why this gives an LUOF is very similar. Given a nonzero message vector, the first encoding step gives many coordinates on which it is nonzero, and on these coordinates the small matrices map them to uniform values, which the final step can extract into a single uniform string, as required. 

Now, to get codes where both it and its dual have good distance, we look at the code with systematic encoding map $x \mapsto (x,xA)$, where now $A$ is supported on $k \times r$ matrices. This gives a code $\cC$ of block-length $n:=k+r$ and dimension $k$, and if $A$ is fast, then so is the map $x \mapsto (x,xA)$. Additionally, it is easy to verify that the code defined by the encoding map $y \mapsto (-yA^\top,y)$ is dual to $\cC$, and the transposition principle guarantees that computing $y \mapsto -yA^\top$ is also fast. Furthermore, we show that if $A$ is LUOF, then so is $A^\top$: the crucial ingredient in this argument is Vazirani's XOR Lemma~\cite{ChorGHFRS85,Rao07}. Finally, a simple argument shows then that both $\cC$ and $\cC^\perp$ get $\eps$-close to the GV bound with high probability, as desired. 

\medskip

To consider more sophisticated properties like list-decodability, we observe that it suffices to consider the following generalization of the LUOF property. We say a random matrix $A$ supported on $k \times r$ matrices is $t$-LUOF if, given $t$ \emph{linearly independent} vectors $x_1,\dots,x_t \in \F_q^k$, the vectors $x_1A,\dots,x_tA$ are uniform and independent over $\F_q^r$.\footnote{In this second case, we mean independent in the ``stochastic'' sense.} Again, we observe that if $A$ is a uniformly random matrix, then it has this property for any $t \leq k$ (for $t>k$, the concept is vacuous, as there are no sets of $k+1$ or more linearly independent vectors in $\F_q^k$). 

Our first observation is that, by assuming the first code of the DI construction has sufficiently good distance ($>1-1/q^t$), we can obtain a $t$-LUOF via their construction. The argument is similar to above: now, we are given $t$ linearly independent vectors $x_1,\dots,x_t$, and we need to argue that in many coordinates we still have linearly independent vectors (recall that each coordinate of this first encoding step is a length-$\beta$ vector). This indeed follows assuming the distance is sufficiently large. Then, on these good coordinates, the inner matrices map the linearly independent vectors to independent, uniform values, which are then extracted in the final step. 

Having obtained thus a $t$-LUOF $A$, we can then consider the (random) code $\cC$ with encoding map $x \mapsto (x,xA)$, along with its dual defined by the encoding $y \mapsto (-yA^\top,y)$. That we can then prove these codes achieve sophisticated properties like list-decoding then follows from the theory of local similarity, as developed in recent works~\cite{MRRSW24,GM22,MRSY25}. Morally speaking, one observes that if one wishes to prove a code is $(\rho,L)$-list-decodable, then one is interested in computing quantities like $\Pr[x_1,\dots,x_{L+1} \in \cC]$, where $x_1,\dots,x_{L+1}$ are all within distance $\rho$ of some vector $z$. Suppose one has an argument that shows uniformly random linear codes $\cC'$ are with high probability $(\rho,L)$-list-decodable, where $\cC'$ is defined as the kernel of a uniformly random $(n-k) \times n$ matrix. In this case, we know that $\Pr[x_1,\dots,x_{L+1} \in \cC'] = q^{-(n-k)\dim(\mathrm{span}\{x_1,\dots,x_{L+1}\})}$. One can verify that for the random $\cC$ defined as above, we have $\Pr[x_1,\dots,x_{L+1} \in \cC'] = q^{-(n-k)\dim(\mathrm{span}\{x_1,\dots,x_{L+1}\})}$, assuming $A$ is $(L+1)$-LUOF. That is, if we take $t=L+1$, then we get list-decodability for $\cC$ (or, more precisely, whatever list-decodability is known to hold for $\cC'$ whp, also holds for $\cC$ whp). Also, generalizing the reasoning above we can show that if $A$ is $t$-LUOF, then so is $A^\perp$, in which case we can derive analogous list-decodability guarantees for $\cC^\perp$.

\subsection{Related Work} \label{subsec:related-work}

While we have already mentioned the papers most directly connected to our work, we now mention a few more relevant related works. 

Work by G\'al, Hansen, Kouck\'y, Pudl\'ak and Viola~\cite{GHKPV13} studies circuit sizes for encoding maps of asymptotically good codes over the binary alphabet. However, their model allows for gates with \emph{unbounded} fan-in (whereas we allow only fan-in $2$), and additionally they insist of small depth (constant, or $\log^* n$), whereas we allow for larger depth (say, $O(\log n)$). 

Another related work by Li~\cite{li2026secretsharingsuperconcentrator} considers the \emph{unbounded} arithmetic circuit complexity (i.e., gates are arbitrary and have unbounded fan-in) of secret-sharing schemes. The main result is a demonstration that the graph underlying the circuit computing the shares for such a scheme must be a \emph{superconcentrator}. Based on this fact, lower bounds on the size of such a sharing circuit are derived.

A final related work is due to Fan, Li and Yang~\cite{FanLY22}, which builds a pairwise independent hash function with optimal size of $\approx 2n$ (where $n$ is the input length) in the $B_2$ circuit model (i.e., Boolean circuits with arbitrary fan-in 2 gates).\footnote{They obtain additional bounds for other circuit classes.} Unlike the IKOS construction (and ours), their construction only yields highly compressive mappings (that is, the output length is $\approx n^{0.1}$), it is not \emph{perfectly} pairwise independent (only statistically close), and only achieves linear size when the circuit can depend on the key (rather than having a single universal circuit of linear size that takes both the key and the input).

\section{Preliminaries} \label{sec:prelims}
We use $\N$ to denote the set of positive integers, and for $n \in \N$ we abbreviate $[n]:=\{1,2,\dots,n\}$. Throughout, $q$ denotes a prime power, and $\F_q$ a finite field with $q$ elements. By default, we view vectors $x \in \F_q^n$ as row vectors. We will also use tuple notation to denote concatenation: given $x \in \F_q^n$ and $y \in \F_q^m$, $z=(x,y)$ denotes the length $n+m$ vector such that $z_i = x_i$ for $i \in [n]$, and $z_i=y_{i-n}$ for $i \in \{n+1,\dots,n+m\}$. 

For a distribution $\cD$ over a finite set $U$ (which itself is a function $\cD:U \to [0,1]$ satisfying $\sum_{u \in U}\cD(u)=1$), we denote by $X \sim \cD$ a random variable distributed according to $\cD$, i.e., for all $u \in U$, $\Pr[X=u]=\cD(u)$. For a finite set $S$, we write $X \sim S$ to denote that $X$ is a random variable distributed uniformly at random over the set $S$, i.e., $\Pr[X=u] = \frac{1}{|S|}$ if $u \in S$, and $0$ otherwise. 

For vectors $x,y \in \F_q^n$, their \emph{inner-product} is defined as $\langle x,y\rangle := \sum_{i=1}^n x_iy_i \in \F_q$. We recall that the inner-product thus defined is commutative ($\langle x,y\rangle = \langle y,x\rangle$), and that additionally it is linear in both arguments, e.g., $\langle \alpha x + \beta y,z\rangle = \alpha \langle x,z\rangle + \beta \langle y, z\rangle$, where also $z \in \F_q^n$ and $\alpha,\beta \in \F_q$. For a subspace $V \leq \F_q^n$, we denote its \emph{dual space} by $V^\bot := \{ u \in \F_q^n : \forall v \in V,~\langle u,v\rangle = 0\}$, which we recall is also a subspace and has dimension $\dim(V^\bot) = n - \dim(V)$, and that additionally $(V^\bot)^\bot = V$. Lastly, we recall that given a matrix $M \in \F_q^{a \times b}$ and vectors $x \in \F_q^a$ and $y \in \F_q^b$, we always have $\langle x M,y\rangle = \langle x,yM^\top\rangle$, where $M^\top \in \F_q^{b \times a}$ is the transpose of $M$.

\subsection{Coding Theory} \label{subsec:coding-theory}

An \emph{$\F_q$-linear code} $\cC$ is an $\F_q$-linear subspace of $\F_q^n$ for some $n \in \N$. The \emph{block-length} of $\cC$ is the value $n$. If $\dim(\cC) = k$, then $\cC$'s \emph{rate} is $R(\cC):=\frac kn$. Given two vectors $x,y \in \F_q^n$, the \emph{(relative) Hamming distance} between them is $d(x,y):=\frac{1}{n}\;|\{i \in [n]:x_i \neq y_i\}|$. Additionally, for a vector $x \in \F_q^n$ its \emph{(relative) Hamming weight} is $\wt(x) := \frac{1}{n}\;|\{i \in [n] : x_i \neq 0\}| = \frac{1}{n}\;|\mathrm{supp}(x)|$, where we've also defined the \emph{support} of a vector as the set of nonzero coordinates. Note that $\wt(x) = d(x,0)$ and that $d(x,y) = \wt(x-y)$. Because of this latter equality, for any $\F_q$-linear code we have that the \emph{(relative) minimum distance} $\delta(\cC):=\min\{d(c,c'):c,c' \in \cC,c \neq c'\}$ is also equal to $\min\{\wt(c):c \in \cC \setminus \{0\}\}$. Lastly, for a linear code $\cC$, we refer to $\cC^\perp$ as its \emph{dual code}.

Every $\F_q$-linear code of dimension $k$ admits a generator matrix $G \in \F_q^{k \times n}$, which is a rank-$k$ matrix such that $\cC = \{xG:x \in \F_q^k\}$. It additionally admits a parity-check matrix $H \in \F_q^{(n-k) \times n}$, which is a rank-$n-k$ matrix such that $\cC = \{x \in \F_q^n: Hx^\top=0\}$. It follows that if $H$ is a parity-check matrix for $\cC$, then $H$ is a generator matrix for $\cC^\perp$. Additionally, the matrices satisfy $GH^\top = 0$. If $G$ is of the form $G = \begin{bmatrix} I_k ~|~ A \end{bmatrix}$ where $I_k \in \F_q^{k \times k}$ is the $k \times k$ identity matrix and $A \in \F_q^{k \times (n-k)}$, then $G$ is said to be \emph{systematic}. Note that if $G=\begin{bmatrix} I_k ~|~ A \end{bmatrix}$ is a systematic generator matrix for $\cC$, then $H = \begin{bmatrix} -A^\top ~|~ I_{n-k} \end{bmatrix}$ is a parity-check matrix for $\cC$, as 
\[
    GH^\top  = \begin{bmatrix} I_k ~|~ A \end{bmatrix} \begin{bmatrix} -A \\ I_{n-k} \end{bmatrix} = -A + A = 0 \ .
\]
We will also be interested in codes where the underlying alphabet is in fact a vector space over $\F_q$. Specifically, for $\beta \in \N$ we will refer to an $\F_q$-subspace $\cC \leq \left(\F_q^\beta\right)^n$ as an \emph{$\F_q$-additive code over $\F_q^\beta$}. The block-length is still $n$, while the rate of such a code is now $\frac{k}{\beta n}$, and we define its minimum distance with respect to the alphabet $\F_q^\beta$. That is, given a codeword $c \in \cC$, writing it as $c = (c^{(1)},\dots,c^{(n)})$ with each $c^{(i)} \in \F_q^\beta$, we define $\wt_\beta(c) := \frac{1}{n} |\{i \in [n] : c^{(i)} \neq 0\}|$, and then $\delta(\cC) := \min \{\wt_\beta(c):c \in \cC\setminus \{0\}\}$. Note that all $\F_q$-additive codes over $\F_q^\beta$ admit a generator matrix $G$ of size $k \times \beta n$. 

We often deal not with a single code $\cC$, but rather with an infinite family of codes $\{\cC_i\}_{i \geq 1}$ for an increasing sequence of block-lengths $n_i\in\N$, each with dimension $\{k_i\}_{i \geq 1}$ and distance $\{\delta_i\}_{i \geq 1}$. The rate of the family of codes can then be defined as $R:=\liminf_{i\to\infty} k_i/n_i$ and the relative distance as $\delta:=\liminf_{i\to\infty} \delta_i$. We call a family of codes \textit{asymptotically good} whenever both $R>0$ and $\delta>0$. We assume the field $\F_q$ is the same within a family. 

The $q$-ary Gilbert-Varshamov bound gives a lower bound on the asymptotic relative distance $\delta$ given the rate $R$ for a family of codes over $\F_q$. In particular, it represents the minimum distance achieved by uniformly random (linear) codes (defined, e.g., by sampling a uniformly random generator matrix). Its definition makes use of the \emph{$q$-ary entropy function} $h_q:[0,1-1/q] \to [0,1]$, which itself is defined as
\[
	h_q(x):=x\log_q(q-1) + x\log_q\frac1x + (1-x)\log_q\frac{1}{1-x} 
\]
for $x \in (0,1-1/q]$ and $h_q(0)=0$ at the left endpoint. This function increases continuously on its domain, and therefore admits a continuous inverse $h_q^{-1}:[0,1] \to [0,1-1/q]$. The \emph{Gilbert-Varshamov (GV) bound} states that there exist codes of distance $\geq \delta$ and rate $\geq 1-h_q(\delta)$. In other words, codes of rate $R$ can achieve distance $\geq h_q^{-1}(1-R)$. If a rate $R$ code achieves distance $h_q^{-1}(1-R-\eps)$ for some $\eps>0$, we will say that it is \emph{$\eps$-close to the GV bound}. We recall that uniformly random codes are $\eps$-close to the GV bound with probability $\geq 1-q^{-\eps n}$. 

Lastly, we recall the concept of \emph{list-decodability}. For $\rho \in (0,1-1/q)$ and $L \in \N$, we say that a code $\cC \subseteq \Sigma^n$ (for an arbitrary finite alphabet $\Sigma$) is $(\rho,L)$\emph{-list-decodable} if, for all $z \in \Sigma^n$, we have
\[
    |\{c \in \cC: d(c,z) \leq \rho\} | \leq L \ .
\]
It is known codes of rate $R$ can achieve decoding radius $\rho$ so long as $\rho < h_q^{-1}(1-R)$. More precisely, there exist codes of rate $1-h_q(\rho)-\eps$ that are $(\rho,O(1/\eps))$-list-decodable; however, any code of rate $1-h_q(\rho)+\eps$ is not $(\rho,L)$-list-decodable for any $L \leq q^{o(n)}$. For this reason, $1-h_q(\rho)$ is called \emph{list-decoding capacity}. We will say that a code is \emph{$\eps$-close to list-decoding capacity} if it has rate $1-h_q(\rho)-\eps$. It is known that there exist codes $\eps$-close to list-decoding capacity with list-size $O(1/\eps)$; in fact, random linear codes achieve this with high probability~\cite{GHK11, LW20}. Furthermore, it is not known if codes $\eps$-close to list-decoding capacity exist with list-size $o(1/\eps)$; thus, hoping for list-size $O(1/\eps)$ is a reasonable target. 

\subsection{Computational Model} \label{subsec:computational}
We will be interested in encoding functions and hash functions that admit ``fast'' algorithms. This refers by default to having a uniform family of linear-size arithmetic circuits, but can also refer to other computational models in which Spielman's error correcting codes can be implemented in linear time. See~\cite{Spielman} for discussion. 

More concretely, our default complexity measure refers to the size of an arithmetic circuit over a finite field $\F_q$. Wires carry field elements into gates. Gates have two input wires and output the addition, multiplication or subtraction of the incoming field elements. 
We allow gates to have arbitrary fan-out. We also allow gates with a fan-in of 0, which represent either an input variable, a constant scalar, or (for randomized circuits) an independently uniformly random value from $\F_q$. We will typically consider {\em linear} circuits in which one of the inputs for each multiplication gate is a scalar

The {\em size} of the circuit is defined as the number of wires, the {\em depth} as the length of the longest path from an input to an output, and the {\em algebraic degree} as the highest (total) degree of an output as an $\F_q$-polynomial in the inputs. 

We consider infinite families of circuits $C_k$ over $\F_q$ where $C_k$ has $k$ input variables. We consider uniform families, which means there is a polynomial time algorithm that maps $1^k$ to a description of $C_k$. We call a function $f:\F_q^k\to \F_q^r$ \emph{fast} if there exists a uniform circuit family $C_k$ which computes $f$, where each $C_k$ has size $O(k)$. We will sometimes apply this notion also to families of {\em randomized} functions computed by randomized circuits.  Finally, for a family of {\em linear} functions, we assume by default that the family of linear-size circuits are {\em linear} in the sense defined above. 

\subsection{Hash Functions}
First, we recall the definition of a hash function with bounded independence. 

\begin{definition} [$t$-wise independent hash function] \label{defn:t-wise-hash}
	Let $q,k,r \in \N$ and let $\cH$ denote a distribution over functions from $[q]^k \to [q]^r$. We say that $\cH$ is a \emph{$t$-wise independent hash function family} if, for all $x_1,\dots,x_t \in [q]^k$ \emph{distinct}, we have, over $H \sim \cH$, 
	\[
		(H(x_1),\dots,H(x_t)) \sim [q]^{rt} \ .
	\]
	In other words, for all $y_1,\dots,y_t \in [q]^r$ it holds that
	\[
		\Pr\left[\bigwedge_{i=1}^t H(x_i)=y_i\right] = q^{-rt} \ .
	\]
\end{definition}

We recall the following classical construction of $t$-wise independent hash functions, based on degree $\leq (t-1)$-polynomials~\cite{CarterWegman1979Universal,Vadhan2012Pseudorandomness}.

\begin{proposition} [Construction of $t$-wise independent hash functions]\label{prop:poly-t-wise-hash}
	Let $q$ be a prime power, let $\beta \in \N$ with $q^\beta \geq t$, and let $\cH$ denote the uniform distribution over functions $\F_{q^\beta} \to \F_{q^\beta}$ of the form $h(x) = \sum_{i=0}^{t-1}h_i x^i$. That is, the domain and range are interpreted as the finite field with $q^\beta$ elements, and the sampled functions are all degree $\leq (t-1)$ polynomials. Then $\cH$ is a $t$-wise independent hash function family. 
\end{proposition}

Observe that sampling according to $\cH$ as defined above requires $t\beta \log_2 q$ bits of uniform randomness. By considering the entropy of the resulting distribution $(H(x_1),\dots,H(x_t))$, this is optimal, as the uniform distribution over $\F_q^{\beta \cdot t}$ has entropy $\beta t \log_2(q)$, and by the data-processing inequality entropy cannot increase. 

In particular, if $\cH$ only ever output functions represented by polynomials of degree $\leq t'-1$ with $t'<t$, then $\cH$ would be supported on a set of size $\leq q^{\beta t'}$, and therefore $\cH$ would have entropy at most $\log q^{\beta t'} = \beta t' \log(q)$, a contradiction to the above. Hence, guaranteeing that every hash function has degree $\leq t-1$ as a polynomial over $\F_{q^\beta}$ is optimal. 

However, we can also think of the functions as being maps from $\F_q^\beta \to \F_q^\beta$ -- that is, inputs and outputs are length-$\beta$ vectors with coordinates in $\F_q$ -- then the functions can be chosen to have smaller degree. That is, one can fix an $\F_q$-linear isomorphism $\F_q^\beta \to \F_{q^\beta}$ such that, under this identification, the function
\[
	F : \F_{q^\beta} \to \F_{q^\beta}, ~~ x \mapsto \sum_{i=0}^{t-1}F_i \cdot x^i 
\]
is represented by a function $G = (G_1,\dots,G_\beta):\F_q^\beta \to \F_q^\beta$ where each $G_j$ has degree $\leq (q-1)\log_q(t)$~\cite[Section~2]{BC12}. More precisely, one can bound the degree of each $G_j$ by
\[
    \max\{|i|_q:F_i \neq 0\} \ ,
\]
where we've defined $|i|_q$ as $\sum_{j=0}^{\beta-1}x_j$ where $i = \sum_{j=0}^{\beta-1}x_jq^j$. In other words, it is the sum of the coordinates in the base-$q$ expansion of $i$. 

Thus, we can have $t$-wise independent hash functions with degree $(q-1)\log_q(t)$: if $i \leq t-1$, one can verify that $|i|_q \leq (q-1)\log_q(t)$. When $q=2$, this bound is just $\log_2(t)$. We record this fact below. 

\begin{lemma} \label{lem:low-degree-hash}
	Let $q$ be a prime power and let $\beta \in \N$ with $q^\beta \geq t$. There exists a $t$-wise independent hash function family $\cH$ of functions from $\F_q^\beta \to \F_q^\beta$ such that every function $h \in \mathrm{supp}(\cH)$ has algebraic degree $\leq (q-1)\log_q(t)$. 
\end{lemma}

\subsection{The $q$-ary XOR Lemma} \label{subsec:fourier-prelims}

Finally, we record a useful lemma, sometimes referred to as Vazirani's XOR Lemma~\cite{ChorGHFRS85}. We will use the generalization to larger fields, for which a proof can be found in Rao's exposition of Bourgain's 2-source extractor~\cite[Lemma~4.2]{Rao07}. In fact, we just need the $\epsilon=0$ version of that statement, which results in the following.

\begin{lemma} \label{lem:inner-products-uniform}
    Let $x$ be a random vector distributed over $\F_q^n$. Then $x \sim \F_q^n$ if and only if for all $\xi \in \F_q^n\setminus\{0\}$, $\langle \xi,x\rangle \sim \F_q$.
\end{lemma}

In words: in order to test if a random vector is uniform over $\F_q^n$, it suffices to check that its inner-product with each (fixed) non-zero vector is uniform over $\F_q$.

\section{Constructions of Fast Codes with Fast Duals}
In this section we concern ourselves with the construction of a ``fast good code" with a ``fast good dual,'' as coined by \cite{BR26}. These are pairs of linear codes that are dual to each other, that are both fast, and that are both asymptotically good with a rate-distance tradeoff at the GV-bound. Applications of such a pair of objects include  protocols for information-theoretic secure multiparty computation (MPC) and encrypted matrix vector products with asymptotically optimal circuit size, as we discuss further in \Cref{sec:apps}.

Recalling from the introduction, such a pair of codes was first given by \cite{BR26}, but their construction had a number of downsides, like being restricted to binary alphabets and both codes needing to have rate $1/2$. We present a new such construction  in \Cref{sec:first_dual} which amends the five drawbacks described in the introduction. The construction relies on a fast {\em linear uniform-output family} (LUOF), a distribution over matrices $A$ such that $xA$ is uniform for any nonzero $x$ (see~\Cref{defn:luof}). The properties of this new construction are summarized in the following theorem. When we refer to a map as being \emph{systematic} we mean that it maps an input vector $x$ to an output vector of the form $(x,y)$ or $(y,x)$, i.e. the output vector contains the input vector.

\begin{theorem}\label{thm:first_dual}
	Let $q$ be a prime power and fix $R,\epsilon \in (0,1)$ with $\eps < \min\{R,1-R\}$. There exist uniform families of randomized linear arithmetic circuits
    $\{C_{i}:\F_q^{R_in_i} \to \F_q^{n_i}\}_{i \in \N}$ and $\{C_i^\perp:\F_q^{(1-R_i)n_i} \to \F_q^{n_i}\}_{i \in \N}$, where $(n_i)_{i \in \N}$ is an increasing sequence of positive integers and $(R_i)_{i \in \N}$ is a sequence of real numbers with each $R_in_i \in \N$, satisfying the following properties. 
	\begin{itemize}
        \item \textbf{Rate:} $\lim_{i \to \infty}R_i = R$ (and hence $\lim_{i \to \infty}(1-R_i) = 1-R$).
        \item \textbf{Efficient encodability:} For each $i \in \N$, $C_{i}:\F_q^{R_in_i} \to \F_q^{n_i}$ is of size $O(n_i)$ with $O(n_i)$ uniformly random field elements. The map is systematic.
        \item \textbf{Dual efficient encodability:} For each $i \in \N$, $C_{i}^\perp:\F_q^{(1-R_i)n_i} \to \F_q^{n_i}$ of size $O(n_i)$ with $O(n_i)$ uniformly random field elements. The map is systematic.  
		\item \textbf{Good distance:} For each $i \in \N$, let $\cC_i$ denote the random code with encoding map determined by $C_i$, where the randomness is over the internal randomness of the circuit. Then $\cC_i$ has rate $R_i$ and distance $h_q^{-1}(1-R_i-\eps)$ except with probability $q^{-\eps n_i}$ (GV-bound).
		\item \textbf{Good dual distance:} For each $i \in \N$, let $\cC_i^\perp$ denote the random code with encoding map determined by $C_i^\perp$, where the randomness is over the internal randomness of the circuit. Then $\cC_i^\perp$ has rate $1-R_i$ and distance $h_q^{-1}(R_i-\eps)$ except with probability $q^{-\eps n_i}$ (GV-bound).
        \item \textbf{Duality:} For each $i\in \N$, the codes $\cC_i$ and $\cC_i^\perp$ are dual to one another, assuming the encoder circuits use the same randomness. 
	\end{itemize}
\end{theorem}

The distance analysis of these codes is much simpler than that of the construction of \cite{BR26}. This gives a better prospect of generalizing to other properties; to prove that our code construction doesn't just generate a ``fast good code" with ``fast good dual," but a code with additional interesting properties. We indeed manage to do so, and in \Cref{sec:second_dual} prove that our construction yields a pair of codes that additionally are $t$-\textit{locally similar} to a random linear code (albeit, at some cost in terms of parameters). Intuitively, this means that any $t$-\textit{local property} of random linear codes also applies to this pair of codes, which are features that can be ruled out by a small, size $t$ subset of vectors. This includes distance ($t=1$) and list-of-$L$-decodability ($t=L+1$). Note that this generalizes the construction from the first subsection, which we leave mostly as a warm-up. 

We conclude with the following result, which, to the best of our knowledge, is the first (randomized) construction of a code achieving list-decoding capacity with fast (\emph{linear-size}) encoding. This holds for both the code we construct \emph{and} its dual. In fact, the codes we construct have the best-known list-size of $L = O(1/\eps)$. 

\begin{theorem}\label{thm:second_dual}
	Let $q$ be a prime power and fix $R,\epsilon \in (0,1)$. There is a $c>0$ such that, assuming $\min\{1-R-c\eps,R-c\eps\}>0$, the following holds. There exist uniform families of randomized linear arithmetic circuits
    $\{C_{i}:\F_q^{R_in_i} \to \F_q^{n_i}\}_{i \in \N}$ and $\{C_i^\perp:\F_q^{(1-R_i)n_i} \to \F_q^{n_i}\}_{i \in \N}$, where $(n_i)_{i \in \N}$ is an increasing sequence of positive integers and $(R_i)_{i \in \N}$ is a sequence of real numbers with each $R_in_i \in \N$, satisfying the following properties. 
	\begin{itemize}
        \item \textbf{Rate:} $\lim_{i \to \infty}R_i = R$ (and hence $\lim_{i \to \infty}(1-R_i) = 1-R$).
        \item \textbf{Efficient encodability:} For each $i \in \N$, $C_{i}:\F_q^{R_in_i} \to \F_q^{n_i}$ is of size $O(n_i)$ with $O(n_i)$ uniformly random field elements. The map is systematic.
        \item \textbf{Dual efficient encodability:} For each $i \in \N$, $C_{i}^\perp:\F_q^{(1-R_i)n_i} \to \F_q^{n_i}$ of size $O(n_i)$ with $O(n_i)$ uniformly random field elements. The map is systematic.  
		\item \textbf{Good list-decodability:} For each $i \in \N$, let $\cC_i$ denote the random code with encoding map determined by $C_i$, where the randomness is over the internal randomness of the circuit. Then $\cC_i$ has rate $R_i$ and is $(h_q^{-1}(1-R_i-c\eps),\lceil1/\eps\rceil)$-list-decodable except with probability $q^{-\Omega(n_i)}$ (list-decoding capacity).
		\item \textbf{Good dual list-decodability:} For each $i \in \N$, let $\cC_i^\perp$ denote the random code with encoding map determined by $C_i^\perp$, where the randomness is over the internal randomness of the circuit. Then $\cC_i^\perp$ has rate $1-R_i$ and is $(h_q^{-1}(R_i-c\eps),\lceil 1/\eps\rceil)$-list-decodable except with probability $q^{-\Omega(n_i)}$ (list-decoding capacity).
        \item \textbf{Duality:} For each $i\in \N$, the codes $\cC_i$ and $\cC_i^\perp$ are dual to one another, assuming the encoder circuits use the same randomness. 
	\end{itemize}
\end{theorem}

\subsection{Fast Good Code with Fast Good Dual}\label{sec:first_dual}
Our construction for a ``fast good code" with ``fast good dual" is based on an object known as a linear uniform output family. This is any matrix $A$ which maps any fixed non-zero input vector to a uniformly random vector in the output space. 

\begin{definition}
\label{defn:luof}
(Linear Uniform Output Family, LUOF)
    Let $A$ be a random matrix distributed over $\F_q^{k \times r}$. We call $A$ a \emph{linear uniform output family (LUOF)} if, for all $x \in \F_q^k\setminus \{0\}$, we have $xA \sim \F_q^r$. 
\end{definition}

Classical constructions of LUOF, such as the family of Toeplitz matrices, require circuits of quasilinear size.
Here we use the fast (i.e., linear-size universal circuit) LUOF construction of Druk and Ishai~\cite{DI14}, which works over any finite field $\F_q$ and for any choice of $k$ and $r$. This LUOF can be used to construct a fast good code: simply use the LUOF $A \in \F_q^{k \times n}$ itself as a generator matrix for a linear code. Since the LUOF is fast, so is the code. The code also is asymptotically good, with rate-distance tradeoff at the GV-bound, which follows in a standard way from the fact that each message is mapped to a uniformly random vector in the output space. We briefly recall this argument as we will use a similar argument for our code constructions. 

The probability the code attains some distance $\delta$ can then (by a union bound) be upper bounded by the probability that any message attains weight at most $\delta$, times the number of messages. Since each message ends up as a uniformly random output vector, the probability it has weight at most $\delta$ is equal to the number of weight at most $\delta$ output vectors divided by the size of the output space: $q^{n \cdot h_q(d/n)-n}$, where $h_q$ is the $q$-ary entropy function. Multiplying by the number of messages $q^{k}$ gives an upper bound on the probability the code attains distance $\delta$ of $q^{n(h_q(\delta) - 1 + k/n)}$, which is $q^{-\Omega(n)}$ (negligible in $n$) whenever $k/n < 1-h_q(\delta)$. This tradeoff between the rate $R:=k/n$ and distance $\delta$ is exactly the $q$-ary GV-bound. 

The fast LUOF of Druk and Ishai thus generates a fast good code. But the authors note that the dual of any code generated by a LUOF likewise is good: it also attains the GV-bound. However, it is not clear what the encoding function of such a dual code looks like, and in particular, whether this dual code is fast. We now give a slightly more complicated code construction which does explicate what the dual code looks like, and indeed lets us guarantee that the dual code is fast. 

Instead of having the primal code be generated by a $k \times n$ LUOF, we suppose that we have a $k \times r$ LUOF called $A$. We then consider the linear maps $x \mapsto (x,xA)$ and $x\mapsto (-xA^\top,x)$. These are implemented by the generator matrices $G=[I_k|A]$ and $H=[-A^\top|I_r]$. Here, $I_k$ and $I_r$ are the $k \times k$ and $r \times r$ identity matrices. Note that $G$ is a $k \times (r+k)$ matrix and $H$ is a $r \times (r + k)$ matrix. Let us write $n:=r+k$ for the block length, so that the primal code generated by $G$ has rate $k/n$ and the dual code generated by $H$ has rate $(n-k)/n$.

We now claim that $G$ is a fast good code with a fast good dual $H$ satisfying the properties listed in \Cref{thm:first_dual}. We spend the rest of this subsection verifying each of the properties listed in this theorem. 

\begin{proof}[Proof of \Cref{thm:first_dual}]
    To start, we can choose $G$ to have any rate $R\in(0,1)$, the codes $G$ and $H$ clearly have systematic encoding maps, and $G$ and $H$ generate dual codes as $H$ is the parity-check matrix of $G$:
    \[
        GH^\top = [I_k~|~A]~\begin{bmatrix}
            -A \\ I_r
        \end{bmatrix} = -I_kA+AI_r = -A+A = 0 \ .
    \]
    The two harder properties that remain are the fact that $G$ and $H$ admit fast encoding maps (arithmetic circuits of linear size in the input) and attain the GV-bound. Starting with the fast encoding maps, recall that that $xG=(x,xA)$ and that $xH=(-xA^\top,x)$. Copying $x$ and applying a minus sign can of course both be done fast. If we sample the LUOF $A$ according to the construction of \cite{DI14}, then computing $xA$ is fast too. From this fact it actually follows that computing $xA^\top$ is fast too:
    
    \begin{lemma}[Transposition principle~\cite{bordewijk1956inter}]\label{lem:fast_tranpose}
        There exists a constant $c_4>0$ such that for every finite field $\F$ the following holds. Let $C$ be an arithmetic circuit over $\F$ of size $t$ which consists of only addition and scalar multiplication gates and computes the function $f(x)=xA$ (where $xA$ is a matrix-vector product). Then there exists an arithmetic circuit $C'$ of size at most $c_4t$ that computes the function $f'(x)=xA^\top$.
    \end{lemma}
    
    All that remains is to show that $G$ and $H$ attain the GV-bound. We start with the code $G$, which we recall encoded a vector $x$ to $(x,xA)$. What we need to show is that the probability that $G$ fails to have distance $h_q^{-1}(1-R-\epsilon)$ for some $\epsilon>0$ is at most $q^{-\epsilon n}$, where we recall that $n$ is the block length of $G$. Since our code is linear, this event is equivalent to $G$ having a codeword of weight at most $h_q^{-1}(1-R-\eps)$.

    There are at most $q^{nh_q(\delta)}$ length $n$ vectors of weight at most $\delta$. Recall the earlier high-level proof that a random linear code attains the GV-bound (in fact, any generator matrix that is a LUOF). There, we could argue that any message vector is mapped to a uniformly random output vector, so that any specific vector $y \in \F_q^n$ is in the code with probability $q^{k-n}$: there are $q^k$ messages, and each has a probability of $q^{-n}$ of being mapped to a some specific output vector. Applying a union bound over all bad codewords (with weight below the target $d$) yields an upper bound on the probability that such a code fails to have distance $\delta$: 

    $$q^{nh_q(\delta)+k-n}=q^{n\left(h_q(\delta)-1+k/n\right)} = q^{n\left(h_q(\delta) - 1 +R\right)} =q^{n\left(1-R-\epsilon - 1 +R\right)}=q^{-\epsilon n} \ ,$$
    where in the last step we substituted in our target distance of $\delta=h_q^{-1}(1-R-\epsilon)$. Now, we cannot directly apply this reasoning to our code, as our code does not map each message to a uniformly random output vector. The first $k$ entries of our output will be a copy of the message, and only the last $r$ entries will be uniformly random. However, it is not hard to see that we can derive the same upper bound of $q^{k-n}=q^{-r}$ on the probability that a given vector $y \in \F_q^n$ is in our code $\cC$. This then immediately yields the GV-bound by the above derivation. 
    
    Recall that $y\in\cC$ whenever $yH^\top=0$ where $H=[-A^\top|I_k]$ is the parity-check matrix of $\cC$. Let us split up $y = (y^{(1)},y^{(2)})$ with $y^{(1)} \in \F_q^k$ and $y^{(2)} \in \F_q^r$. Then we can write  
    $$ \Pr\left[y\in \cC\right] = \Pr\left[yH^\top =0\right] = \Pr\left[y^{(2)} - y^{(1)}A =0\right] = \Pr\left[y^{(2)} = y^{(1)}A \right] \leq q^{-r} \ , $$
    where the last step is because $A$ is a LUOF, so that the output $y^{(1)}A$ is uniformly random over $\F_q^r$ if $y^{(1)} \neq 0$; if $y^{(1)}=0$, then $y^{(2)}\neq 0$ (since $y=(y^{(1)},y^{(2)})\neq 0)$, so then $\Pr\left[y^{(2)} = y^{(1)}A \right]=0$. Thus, in our code we obtain the same (upper bound on the) probability that a given codeword is in our code as we had above for a random linear code. It follows that we obtain the exact same bound on the probability that our code achieves distance $d$, and hence we recover the GV-bound.

    It remains to prove that the dual code, generated by $H$ with the linear map $x \mapsto (-xA^\top,x)$, also attains the GV-bound. To do this, we prove \Cref{lem:transpose} just below which states that if $A$ is a LUOF, then $-A^\top$ is also a LUOF. The proof that $H$ achieves the GV-bound is then analogous to that of $G$, establishing the theorem. 
\end{proof}
    
\begin{lemma} \label{lem:transpose}
    Let $r,k \in \N$. Let $A$ be a random matrix distributed over $\F_q^{k \times r}$. If $A$ is LUOF, then so is $-A^\top$.
\end{lemma}
\begin{proof}
    It is immediate that if $A$ is LUOF, then so is $-A$. Thus, it suffices to prove that if $A$ is LUOF, then so is $A^\top$. That is, we must show that for all $x \in \F_q^r\setminus\{0\}$, we have $xA^\top \sim \F_q^k$. By \Cref{lem:inner-products-uniform}, it suffices to prove that for all $\xi \in \F_q^k\setminus\{0\}$, we have $\langle \xi, xA^\top\rangle \sim \F_q$. Now, $\langle \xi, xA^\top \rangle = \langle \xi A,x\rangle$. Since $\xi \neq 0$ and $A$ is LUOF, we know that $\xi A \sim \F_q^r$, so now by the other direction of \Cref{lem:inner-products-uniform} we have $\langle \xi A, x\rangle \sim \F_q$, as desired. 
\end{proof}

\subsection{Fast List-Decodable Code with Fast List-Decodable Dual}\label{sec:second_dual}
We now generalize the previous construction to yield not just a ``fast good code" with ``fast good dual", but codes that are in addition $t$-\textit{locally similar} to a random linear code. We won't define $t$-local similarity formally, which would require some additional definitions. Instead, we make use of Proposition~II.8 of \cite{MRSY25} (a result which was implicit in earlier works~\cite{MRRSW24,GM22}) which states that a code is $t$-locally similar to a random linear code whenever any selection of $t$ linearly independent vectors are in $\cC$ with probability at most $q^{-n(1-R)t}$, where $R$ is the rate, $q$ is the field size and $n$ is the block length.\footnote{Note that this result requires $\frac{n}{\log_qn} \geq \omega_{n\to\infty}(q^{2t})$, i.e. the block length needs to be large enough compared to $q$ and $t$.} We will establish that (a slight generalization of) our code construction satisfies this property.

As explained in the introduction to this section, this roughly implies that our code shares with random linear codes all features that can be ruled out by size $t$ subsets of vectors. Our goal will be to show that our code construction yields a pair of codes that have good list-decodability. Recall that a code $\cC \subseteq \F_q^n$ is $(\rho,L)$-list-decodable if for all $z \in \F_q^n$, $|\{c \in \cC:d(c,z)\leq \rho\}| \leq L$. Observe that proving a code $\cC$ is $(\rho,L)$-list-decodable amounts to showing that certain sets of vectors of size $L$ are \emph{not} contained in $\cC$: in this specific case, the vectors that must be shown to not lie in $\cC$ are all $L+1$-subsets of some Hamming ball of radius $\rho$. This aligns exactly with the intuitive idea that a local property is one that can be ruled out by a small set of vectors. 

Proposition~II.9 of \cite{MRSY25} makes this precise by showing that any linear code that is $t$-locally similar to a random linear code has good list-decodability. We state it in full below.

\begin{proposition}\label{prop:list_decode}
	Let $q$ be a prime power. Let $k,n,L \in \N$, let $R = k/n$ and let $\eps > 1/(L+1)$. For $n$ sufficiently large compared to $L$, there exists a constant $c>0$ such that the following holds. If $\rho = h_q^{-1}(1-R-c\eps)$ and $\cC \leq \F_q^n$ is a random code of rate $R$ that is $(L+1)$-locally similar to a random linear code, then with probability at least $1-q^{-\eps(n-o_{n\to \infty}(1))}$, $\cC$ is $(\rho,L)$-list-decodable.
\end{proposition}

We note that the maximum radius up to which one can hope to list-decode is $h^{-1}_q(1-R-\eps)$, and existentially at radius $h^{-1}_q(1-R-\eps)$ lists of size $O(1/\eps)$ are sufficient. Thus, if we can prove that our codes are $t$-locally similar to random linear codes, they essentially obtain the best-known list-decodability. We now turn to proving this, thereby establishing \Cref{thm:second_dual}. 

We recall once more that to prove $t$-local similarity it suffices to prove that any choice of $b\leq t$ linearly independent vectors are in the codes with probability at most $q^{-n(1-R)b}$. Our original code construction used a fast LUOF which guaranteed that any fixed non-zero input vector is mapped to a uniformly random output vector. This property is not sufficient for our current purposes, as it doesn't guarantee that multiple distinct inputs are mapped to uniform outputs independently. 

To obtain this guarantee, we generalize the notion of a LUOF. Instead of requiring a single non-zero input to be mapped to a uniform random output, we require $t$ linearly independent inputs to be mapped to uniformly and independent outputs. 

\begin{definition}
    Let $t,r,k \in \N$. Let $A$ be a random matrix distributed over $\F_q^{k \times r}$. We call $A$ a \emph{t-wise linear uniform output family ($t$-LUOF)} if, for all $x_1,x_2,\dots,x_t \in \F_q^k$ linearly independent, we have $(x_1A,x_2A,\dots,x_tA) \sim \F_q^{tr}$. That is, the collection of random variables $x_iA$ for $i \in [t]$ are all uniform over $\F_q^r$, and independent. 
\end{definition}
\begin{remark}
    Note that if a random matrix is $t$-LUOF, then it is also $s$-LUOF for any $s \leq t$. Indeed, if $(x_1A,x_2A,\dots,x_tA) \sim \F_q^{tr}$, then also $(x_1A, x_2A,\dots,x_sA) \sim \F_q^{sr}$.
\end{remark}

In \Cref{sec:t-luof} we will show that the DI construction of a fast LUOF can be generalized to yield a fast $t$-LUOF as well. For now, we show that our code construction, set up with a $t$-LUOF instead of a $1$-LUOF, gives us $t$-local similarity. Or more specifically, we prove that our construction yields a fast list-decodable code with fast list-decodable dual, as described in \Cref{thm:second_dual}.

\begin{proof}[Proof of \Cref{thm:second_dual}]
    We use the codes $G$ and $H$ according to the code construction outlined in \Cref{sec:first_dual}, but instead of using the fast 1-LUOF of \cite{DI14}, we use the fast $t$-LUOF which we construct in \Cref{sec:t-luof}. We recall that proof of \Cref{thm:first_dual} in that same subsection which shows us that these codes are dual, can be of any rates (as long as they sum to 1, of course), and have fast, systematic encoding maps. But now, instead of proving good distance, we need to argue that our codes are $t$-locally similar to random linear codes, for which we recall it suffices to show that any choice of $b\leq t$ linearly independent vectors are in the codes with probability at most $q^{-n(1-R)b}$. Note that $R=k/n$ and $n=k+r$, so that $n(1-R)=n(1-k/n)=n-k=r$. The probability upper bound then becomes $q^{-rb}$. 
    
    We now show that $G$ and $H$ indeed have probability $\leq q^{-rb}$ that any fixed selection of $b$ non-zero and linearly independent vectors is in our code, establishing the required list-decodability by \Cref{prop:list_decode}. This will require not just that $A$ is $t$-LUOF, but that $A^\top$ is $t$-LUOF as well, which we establish in \Cref{lem:t-luof_transpose}. Then, in \Cref{lem:q-b} we show that $G$ satisfies the earlier probability bound. A completely analogous proof establishes this result for $H$, which in turn completed the proof of our theorem. 
\end{proof}

\begin{lemma}\label{lem:t-luof_transpose}
    Let $t,r,k \in \N$. Let $A$ be a random matrix distributed over $\F_q^{k \times r}$ which is $t$-LUOF. Then $A^\top$ is also $t$-LUOF. 
\end{lemma}

\begin{proof}
    We must show that for all $\xi_1,\dots,\xi_t \in \F_q^r$ linearly independent, we have $( \xi_1A^\top,\dots, \xi_t A^\top) \sim \F_q^{tk}$. By \Cref{lem:inner-products-uniform}, it suffices to argue that for any $x \in \F_q^{tk} \setminus \{0\}$, we have $\langle x,(\xi_1 A^\top,\dots, \xi_tA^\top)\rangle \sim \F_q$. Write $x = (x_1,\dots,x_t)$ with each $x_i \in \F_q^k$, then we need to argue 
    \[
        \langle x,(\xi_1 A^\top ,\dots,\xi_tA^\top )\rangle= \sum_{i=1}^t \langle x_i, \xi_i A^\top \rangle  = \sum_{i=1}^t\langle \xi_i,x_iA\rangle
    \]
    is uniform. Fix a basis $\{y_i:i \in [s]\}$ for the space $\mathrm{span}\{x_i:i \in [t]\}$, and note that $s \geq 1$ since we assumed that not all the $x_i$'s are $0$. We can (uniquely) write each $x_i = \sum_{j=1}^s x_{ij}y_j$ with the $x_{ij} \in \F_q$. So then
    \begin{align*}
        \sum_{i=1}^t\langle \xi_i,x_iA\rangle &= \sum_{i=1}^t \left\langle \xi_i,\left(\sum_{j=1}^s x_{ij}y_j\right)A\right\rangle = \sum_{i=1}^t \sum_{j=1}^s x_{ij} \langle \xi_i, y_jA\rangle  = \sum_{j=1}^s \left\langle \sum_{i=1}^t x_{ij} \xi_i, y_jA\right\rangle \ .
    \end{align*}
    Now, since $y_1,\dots,y_s$ are linearly independent, we know that $y=(y_1A,\dots,y_sA) \sim \F_q^{sr}$. Considering now the vector $\xi = \left(\sum_{i=1}^t x_{i1}\xi_i,\dots,\sum_{i=1}^tx_{is}\xi_i\right)$, if we show that it is nonzero, then by \Cref{lem:inner-products-uniform} we will conclude that 
    \[
        \sum_{j=1}^s \left\langle \sum_{i=1}^t x_{ij} \xi_i, y_j A\right\rangle = \langle \xi,y\rangle \sim \F_q \ ,  
    \]
    as required. So, we must show $\xi \neq 0$. That is, we must show that for some $j \in [s]$, we have $\sum_{i=1}^tx_{ij}\xi_j \neq 0$. Since $\xi_1,\dots,\xi_t$ are linearly independent, we have $\sum_{i=1}^t x_{ij}\xi_i=0 \iff x_{1j}=x_{2j} = \cdots = x_{tj}=0$; this implies that we just must show some $x_{ij} \neq 0$. To see this is true, note that some $x_i \neq 0$, and so since the $y_i$'s form a basis that means some $x_{ij} \neq 0$, as desired. 
\end{proof}

\begin{lemma}\label{lem:q-b}
    Let $b,t,r,k \in \N$ with $b \leq t$, and let $A$ be a random matrix distributed over $\F_q^{k \times r}$ which is $t$-LUOF. Consider systematic linear code $\cC := \{(x,xA):x \in \F_q^k\}$, which has block-length $n:=k+r$. Then, for all $x_1,\dots,x_b \in \F_q^n$ linearly independent, we have $\Pr\left[\forall j \in [b], x_j \in \cC\right] \leq q^{-rb}$. 
\end{lemma}

\begin{proof}
    Since $\cC$ is generated by $G = \begin{bmatrix} I_k ~|~ A \end{bmatrix}$, it is checked by the matrix $H = \begin{bmatrix} -A^\top ~|~ I_r\end{bmatrix}$. Note that, for all $x \in \F_q^n$, $xH^\top=x^{(2)}-x^{(1)}A$, where $x = (x^{(1)},x^{(2)})$ with $x^{(1)} \in \F_q^k$ and $x^{(2)} \in \F_q^r$. Hence, 
    \begin{align}
        \Pr\left[\forall j \in [b], x_j \in \cC\right] &= \Pr\left[\forall j \in [b], x_jH^\top =0\right]  \nonumber \\
        &= \Pr\left[\forall j \in [b], x^{(2)}_j-x^{(1)}_jA =0\right] = \Pr\left[\forall j \in [b], x^{(1)}_jA =x^{(2)}_j\right] \label{eq:prob-with-transpose}
    \end{align}
    Let us first consider the case that $x^{(1)}_1,\dots,x^{(1)}_b$ are linearly independent. Then, it follows directly from the fact that $A$ is $t$-wise LUOF (and hence, $b$-wise LUOF, since $b \leq t$) that the vector $(x^{(1)}_1A^\top,\dots,x_b^{(1)}A^\top)$ is distributed uniformly over $\F_q^{br}$, and hence the probability it takes on the value $(x_1^{(2)},\dots,x_b^{(2)})$ is $q^{-rb}$. Hence, 
    \[
        \eqref{eq:prob-with-transpose} = \Pr\left[(x^{(1)}_1A^\top,\dots,x_b^{(1)}A^\top)=(x_1^{(2)},\dots,x_b^{(2)})\right] = q^{-rb} \ .
    \]
    Let us now consider the case that $x^{(1)}_1,\dots,x_b^{(1)}$ are linearly dependent. Fix $j^* \in [b]$ minimal such that $x_{j^*}^{(1)} \in \mathrm{span}\{x_1^{(1)},\dots,x_{j^*-1}^{(1)}\}$ (note that we could have $j^*=1$, where we've used the convention that $\mathrm{span}(\emptyset) =\{0\}$). Thus, there exist multipliers $\lambda_1,\dots,\lambda_{j^*-1} \in \F_q$ such that $x_{j^*}^{(1)} = \sum_{i=1}^{j^*-1} \lambda_i x_i^{(1)}$ (again, if $j^*=1$, this is just the empty sum, which has value $0$). We have
    \[
        \eqref{eq:prob-with-transpose} \leq \Prop{\forall j \in [j^*], x_j^{(1)}A^\top =x_j^{(2)}}
    \]
    We claim the above probability is $0$, which certainly establishes the required bound $\eqref{eq:prob-with-transpose} \leq q^{-rb}$. 
    
    To see this term is $0$, we write 
    \begin{align*}
        &\Prop{\forall j \in [j^*],~ x_j^{(1)}A^\top =x_j^{(2)}}\\ 
        &\qquad = \Prop{x_{j^*}^{(1)}A^\top =x_{j^*}^{(2)}\Big| \forall i < j^*, ~x_i^{(1)}A^\top =x_i^{(2)}} \cdot \Prop{\forall i < j^*, ~x_i^{(1)} A^\top =x_i^{(2)}}
    \end{align*}
    Note that the event $\forall i < j^*, ~x_i^{(1)}A^\top =x_i^{(2)}$ does indeed have nonzero probability: if $j^* \ge 2$, as the vectors $x_1^{(1)},\dots,x_{j^*-1}^{(1)}$ are linearly independent and $A^\top$ is $(j^*-1)$-LUOF, the above argument shows this event has probability $q^{-r(j^*-1)}$, while if $j^*=1$ we condition on the empty event, which has probability $1$. 
    
    Now, conditioned on the event $\forall i < j^*, ~x_i^{(1)}A^\top =x_i^{(2)}$, we have 
    \begin{align*}
        x_{j^*}^{(1)}A^\top  = \left(\sum_{i=1}^{j^*-1} \lambda_i x_i^{(1)}\right)A^\top  = \sum_{i=1}^{j^*-1}\lambda_i ~x_i^{(1)} A^\top = \sum_{i=1}^{j^*-1}\lambda_i x_i^{(2)} \ .
    \end{align*}
    Thus, if $x_{j^*}^{(1)}A^\top =x_{j^*}^{(2)}$, then also $\sum_{i=1}^{j^*-1}\lambda_i x_{i}^{(2)} = x_{j^*}^{(2)}$, which then means that 
    \[
        x_{j^*} = (x_{j^*}^{(1)},x_{j^*}^{(2)}) = \sum_{i=1}^{j^*-1}\lambda_i (x_{i}^{(1)},x_{i}^{(2)}) = \sum_{i=1}^{j^*-1}\lambda_i x_{i} \ .
    \]
    Since the $x_1,\dots,x_b$ were assumed linearly independent, this cannot happen. So the probability is $0$, as desired. 
\end{proof}

\section{Fast Bounded-Independence Functions}

In this section we extend previous results from~\cite{IKOS08,DI14} on hash functions and linear codes by considering a higher independence parameter $t$.

\subsection{Fast $t$-Wise Independent Hash Functions} \label{sec:t-wise-hash} 

In this section we prove that, with minor modifications, the IKOS construction~\cite{IKOS08} can yield $t$-wise independent hash functions. This generalized variant is still fast as long as $t$ is constant. We summarized the IKOS construction for pairwise independent hash functions ($t=2$) in \Cref{subsec:our-techniques}. We now briefly recall that the construction consisted first of an encoding step (determined by code with very good distance over a larger alphabet); then, constant-sized hash functions are applied to the coordinates of the encoding; lastly, an extraction step (implemented by the transpose of a generator matrix for a good code) is applied. 

To generalize this to $t$-wise independence, the first natural change that one must make is that the small, constant-sized hashes $h_1,\dots,h_m$ should now be sampled so as to be $t$-wise independent (where $m$ is the block-length of the first code), which we can naturally do (and, assuming $t=O(1)$, this will only affect the circuit-size by a constant factor). 

Now, recalling the proof of pairwise independence for the IKOS construction, the good distance of the first code guaranteed that, in many coordinates, the encodings of the two distinct inputs will be different. That is, if the code's distance was $1-\eps$, that exactly means that in at least a $1-\eps$ fraction of coordinates, the encodings will differ. 

For $t$-wise independence, we are given $t$ inputs $x_1,\dots,x_t$, and in order to use the $t$-wise independence of the little hashes $h_1,\dots,h_m$, we need many coordinates $\ell \in [m]$ in which \emph{all} the encodings differ. In this case, an averaging argument will establish this holds in at least a $1-\binom{t}2\cdot \eps$ fraction of coordinates. Thus, so long as $\eps>0$ is chosen small enough (which can be done at the cost of increasing the alphabet size for the initial code), the argument goes through.

We now formally prove the construction works. We will use the following codes from \cite{DI14}:

\begin{theorem}[{\cite[Theorem~2]{DI14}}] \label{thm:good-code}
    For every prime power $q$ and $0 < \delta < 1$ there exists $0<\rho<1$, $\beta \in \N$ and a family of $\F_q$-additive codes $\cD_k \leq (\F_q^\beta)^m$ with rate $R = \frac{k}{\beta m} \geq \rho$ and minimum $\beta$-distance $\geq \delta$. Recall this means that, for all $c \in \cD_k \setminus\{0\}$, we have $\wt_\beta(c) \geq \delta$.

    Furthermore, the encoding map $E_{\cD_k}$ can be computed by a uniform family of $O(k)$-size arithmetic circuits over $\F_q$. 
\end{theorem}

\begin{remark} \label{rem:decreasing-dimension}
    \Cref{thm:good-code} above promises codes of rate at least some value $\rho$; however, in some cases we will in fact require codes of a smaller rate, while preserving the same distance and running time. We can easily arrange for this by, say, zeroing out some coordinates from the messages. E.g., if $\cD_k$ has dimension $k$ and we'd instead like a code with dimension $k'<k$, we just append $k-k'$ $0$'s to the message $m \in \F_q^{k'}$ prior to encoding via $E_{\cD_k}$.
\end{remark}

\begin{theorem}\label{thm:t-wise-hash}
    Let $k,r\in \N$ be growing parameters, fix $t \in \N$, and let $q$ be a prime power. There exists an arithmetic circuit $C_{k,r}:\F_q^k \to \F_q^r$ of size $O(k+r)$ with $O(k+r)$ uniformly random field elements such that, over the uniformly random field elements, the circuit implements a $t$-wise independent hash function of degree $(q-1)\log_q(t)$.
\end{theorem}

\begin{proof}

    We first describe the construction. Fix constants $0<\eps,\delta_2<1$ with $1-\binom t2 \cdot \eps +\delta_2>1$, and fix $\beta \in \N$ large enough so that the following hold. Let $\cC_1 \subseteq \left(\F_q^\beta\right)^m$ and $\cC_2 \subseteq \left(\F_q^\beta\right)^m$ be $\F_q$-additive codes of dimension $k$ and $r$, respectively, and with distance $1-\eps$ and $\delta_2$, respectively. Such codes are promised to exist by \Cref{thm:good-code} (along with \Cref{rem:decreasing-dimension} in order to reduce the codes' dimensions, if required). Let $G_1$ and $G_2$ denote generator matrices for $\cC_1$ and $\cC_2$, respectively. Let $C_1 : \F_q^k \to \F_q^{m \cdot \beta}$ be the (uniform) arithmetic circuit of size $O(m)$ implementing $G_1$, and let $C_2:\F_q^{m \cdot \beta} \to \F_q^r$ be the (uniform) arithmetic circuits of size $O(m)$ implementing $G_2^\top$. That such circuits exist follows from \Cref{thm:good-code} and \Cref{lem:fast_tranpose}.

    Lastly, the circuit samples $t \cdot \beta \cdot m= O(k+r)$ uniformly random symbols which we take to describe $m$ $t$-wise independent functions $h_1,\dots,h_m:\F_q^\beta \to \F_q^\beta$. We remark that, thanks to \Cref{prop:poly-t-wise-hash}, we can indeed sample such constant-sized $t$-wise hash functions. We imagine composing these functions $h_i$ into the map
    \[
        h : \left(\F_q^\beta\right)^m \to \left(\F_q^\beta\right)^m ~ , ~~~~ h\left(x^{(1)},\dots,x^{(m)}\right) = \left(h_1(x^{(1)}),\dots,h_m(x^{(m)})\right) \ .
    \]
    As each $h_\ell$ can be implemented by a uniform, randomized arithmetic circuit of size $O_{\beta,t}(1)$, $h$ can be implemented by a uniform, randomized arithmetic circuit of size $O(m)$; call this circuit $B$. Note that $B$ then requires $O(k+r)$ uniformly random field elements. Defining the circuit $C_{k,r}$ now as $C_2 \circ B \circ C_1$, it is immediate that it is uniform and of size $O(m) = O(k+r)$, with $O(k+r)$ uniformly random field elements. Additionally, by \Cref{lem:low-degree-hash} we can enforce that all the $h_\ell$'s have algebraic degree $\leq (q-1)\log_q(t)$, and so since the other steps are linear, the composite function also has algebraic degree $\leq (q-1)\log_q(t)$. 

   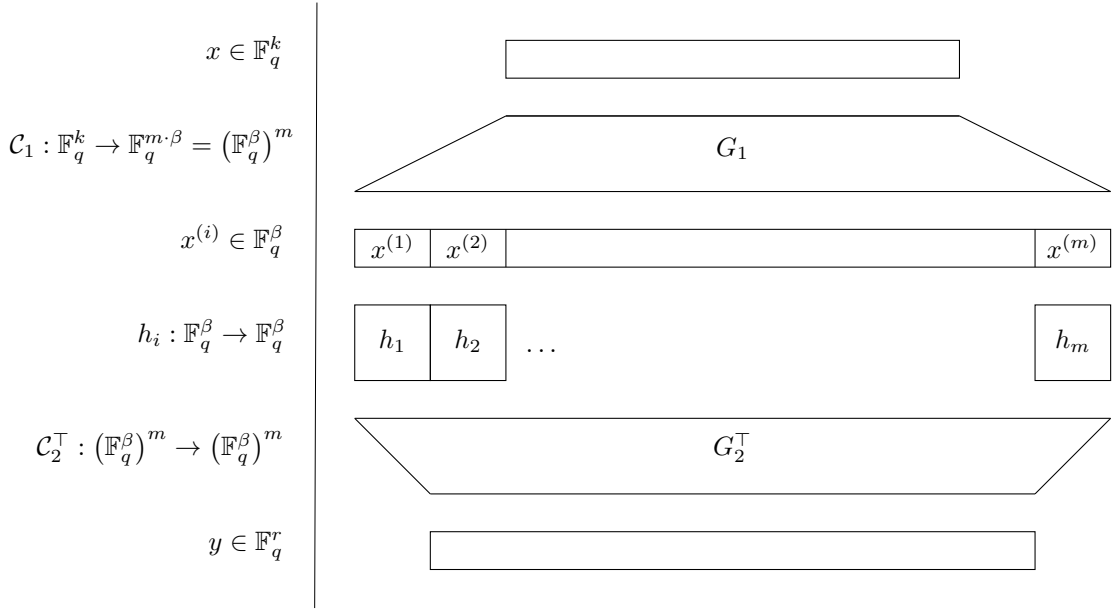
\begin{figure}[!h]
        \centering
        \begin{tikzpicture}[scale=1]
          \coordinate (A) at (0.00, 1.00);
          \coordinate (B) at (1.00, 0.00);
          \coordinate (C) at (9.00, -0.00);
          \coordinate (D) at (10.00, 1.00);
          \coordinate (E) at (0.00, 4.00);
          \coordinate (F) at (2.00, 5.00);
          \coordinate (G) at (8.00, 5.00);
          \coordinate (H) at (10.00, 4.00);
          \coordinate (I) at (5.00, 4.30);
          \coordinate (J) at (5.00, 0.30);
          \coordinate (K) at (0.50, 1.75);
          \coordinate (L) at (1.50, 1.75);
          \coordinate (M) at (2.50, 1.75);
          \coordinate (N) at (9.50, 1.75);
          \coordinate (O) at (-1.46, 5.50);
          \coordinate (P) at (-2.68, 4.25);
          \coordinate (W) at (9.50, 3.00);
          \coordinate (X) at (1.50, 3.00);
          \coordinate (Y) at (0.50, 3.00);
          \coordinate (Z) at (-1.62, 3.00);
          \coordinate (AA_1) at (-2.57, 0.25);
          \coordinate (AA) at (-1.89, 1.75);
          \coordinate (AB) at (-1.44, -1.00);
          \coordinate (Q) at (1.00, 3.50);
          \coordinate (R) at (1.00, 3.00);
          \coordinate (S) at (2.00, 3.50);
          \coordinate (T) at (2.00, 3.00);
          \coordinate (U) at (9.00, 3.50);
          \coordinate (V) at (9.00, 3.00);
          \coordinate (AC) at (-0.50, 6.50);
          \coordinate (AD) at (-0.53, -1.51);
        
          \draw (A) -- (B);
          \draw (B) -- (C);
          \draw (D) -- (C);
          \draw (A) -- (D);
          \draw (0,1.5) rectangle (1,2.5);
          \draw (1,1.5) rectangle (2,2.5);
          \draw (9,1.5) rectangle (10,2.5);
          \draw (E) -- (F);
          \draw (F) -- (G);
          \draw (F) -- (G);
          \draw (G) -- (H);
          \draw (E) -- (H);
          \draw (2,5.5) rectangle (8,6);
          \node[above] at (I) {$G_1$};
          \node[above] at (J) {$G_2^\top$};
          \node[above] at (K) {$h_1$};
          \node[above] at (L) {$h_2$};
          \node[above] at (M) {$\dots$};
          \node[above] at (N) {$h_m$};
          \node[above] at (O) {$x \in \mathbb{F}_q^k$};
          \draw (0,3) rectangle (10,3.5);
          \node[above] at (P) {$\cC_1 : \mathbb{F}_q^k \to \mathbb{F}_q^{m\cdot \beta} = \left(\mathbb{F}_q^\beta\right)^m$};
          \node[above] at (W) {$x^{(m)}$};
          \node[above] at (X) {$x^{(2)}$};
          \node[above] at (Y) {$x^{(1)}$};
          \node[above] at (Z) {$x^{(i)}\in \mathbb{F}_q^\beta$};
          \node[above] at (AA_1) {$ \cC_2^\top : \left(\mathbb{F}_q^\beta\right)^m \to \left(\mathbb{F}_q^\beta\right)^m$};
          \node[above] at (AA) {$ h_i : \mathbb{F}_q^\beta \to \mathbb{F}_q^\beta$};
          \draw (1,-1) rectangle (9,-0.5);
          \node[above] at (AB) {$y \in \mathbb{F}_q^r$};
          \draw (Q) -- (R);
          \draw (S) -- (T);
          \draw (U) -- (V);
          \draw[ultra thin] (AC) -- (AD);
        \end{tikzpicture}
        \caption{A schematic of the IKOS construction for $t$-wise independent hash functions. The functions $h_i$ are constant size $t$-wise independent hash functions which we sample. $\cC_1$ is a code with generator matrix $G_1$ and $\cC_2$ is a code with generator matrix $G_2$.}
        \label{fig:ikos_construction}
    \end{figure}

    Let us write $H:\F_q^{k} \to \F_q^r$ to denote the function defined by the circuit $C_{k,r}$. We refer the reader to \Cref{fig:ikos_construction} which gives a schematic of $H$. Recall that $H$ takes $O(k+r)$ uniformly random field elements. What we want to prove is that, over these random values (which sample the $m$ $t$-wise independent hash functions $h_1,\dots,h_m:\F_q^\beta \to \F_q^\beta$), $H$ is a $t$-wise independent hash function family. That is, given distinct input vectors $x_1,\dots,x_t \in \F_q^k$, over the randomness we have that
    \[
        z:=(H(x_1),\dots,H(x_t)) \sim \F_q^{rt} \ .
    \]
    To do this, by \Cref{lem:inner-products-uniform} it suffices to consider any $y \in \F_q^{tr} \setminus \{0\}$, and to prove that $\langle y,z\rangle \sim \F_q$. Write $y = (y_1,\dots,y_t)$ with each $y_i \in \F_q^r$, and let $h : \left(\F_q^{\beta}\right)^m \to \left(\F_q^{\beta}\right)^m$ be the function determined by sampling the $m$ individual $t$-wise independent hash function $h_i$. Then
    \begin{align}
        \langle z,y\rangle = \sum_{i=1}^t \left\langle y_i , H(x_i)\right\rangle = \sum_{i=1}^t \left\langle y_i , h(x_iG_1)G_2^\top\right\rangle = \sum_{i=1}^t \left\langle y_iG_2 , h(x_iG_1)\right\rangle \label{eq:t-wise-to-continue}\ .
    \end{align}
    Let $c_i = x_iG_1$ and $d_i = y_i G_2$, and note that $c_1,\dots,c_t$ are distinct codewords in $\cC_1$ and $d_1,\dots,d_t$ are codewords in $\cC_2$ which are not all $0$. Without loss of generality, assume $d_1 \neq 0$. Write $c_i = (c_i^{(1)},\dots,c_i^{(m)})$ and $d_i = (d_i^{(1)},\dots,d_i^{(m)})$ with each $c_i^{(\ell)},d_i^{(\ell)} \in \F_q^\beta$. We can now write 
    \[
        \eqref{eq:t-wise-to-continue} = \sum_{\ell=1}^m \sum_{i=1}^t\left\langle d_i^{(\ell)}, h_\ell(c_i^{(\ell)})\right\rangle
    \]
    
    We claim there exists $\ell^* \in [m]$ for which 
    \begin{itemize}
        \item[(a)] $d_1^{(\ell^*)} \neq 0$, and
        \item[(b)] the vectors $c_1^{(\ell^*)},\dots,c_t^{(\ell^*)} \in \F_q^\beta$ are \emph{distinct}. 
    \end{itemize}
    We establish this claim now. Since $\cC_2$ has distance $\geq \delta_2$, condition (a) holds for at least a $\delta_2$ fraction of $\ell \in [m]$. As for the second condition, consider the $\binom t2$ codewords $c_{ij} = c_i - c_j$ for $1 \leq i < j \leq t$. Since $\cC_1$ has distance $1-\eps$, we have that $c_{ij}^{(\ell)} \neq 0$ for at least a $1-\eps$ fraction of $\ell \in [m]$. Thus, we have $c_{ij}^{(\ell)}=0$ for at most an $\eps$-fraction of $\ell \in [m]$, implying $c_{ij}^{(\ell)}=0$ for \emph{some} $1\leq i < j \leq t$ for at most a $\binom t2 \cdot \eps$ fraction of $\ell \in [m]$. Therefore, condition (b) holds for at least a $1-\binom t2\cdot \eps$ fraction of $\ell \in [m]$. Since $1-\binom t2 \cdot \eps + \delta_2>1$ by assumption, by the pigeonhole principle some $\ell \in [m]$ must exist satisfying conditions (a) and (b).

    So, fix such an $\ell^*$. Since the $h_1,\dots,h_m$ are sampled independently, if we prove that 
    $\sum_{i=1}^t\left\langle d_i^{(\ell^*)}, h_{\ell^*}(c_i^{(\ell^*)})\right\rangle\sim\F_q$, then additionally $\sum_{\ell=1}^m \sum_{i=1}^t\left\langle d_i^{(\ell)}, h_\ell(c_i^{(\ell)})\right\rangle \sim \F_q$, as required. To this end, without loss of generality assume there is $1 \leq s\leq t$ such that $d_i^{(\ell)}\neq 0$ iff $i \leq s$. So then 
    \[
        \sum_{i=1}^t\left\langle d_i^{(\ell^*)}, h_\ell(c_i^{(\ell^*)})\right\rangle = \sum_{i=1}^s\left\langle d_i^{(\ell^*)}, h_{\ell^*}(c_i^{(\ell^*)})\right\rangle \ .
    \]
    Since the vectors $c_1^{(\ell^*)},\dots,c_s^{(\ell^*)}$ are distinct (and $s\leq t$), the $t$-wise independence of $h_{\ell^*}$ guarantees the values $h_{\ell^*}(c_1^{(\ell^*)}),\dots,h_{\ell^*}(c_s^{(\ell^*)})$ are independent and uniform, which guarantees that $\sum_{i=1}^s\left\langle d_i^{(\ell^*)}, h_{\ell^*}(c_i^{(\ell^*)})\right\rangle \sim \F_q$, as required. 
\end{proof}

\begin{remark} \label{rmk:eps-t-hash}
    In \Cref{thm:t-wise-hash}, we argued that we can achieve $O(k+r)$ circuit size for constant values of $t$. Naturally, the constant hidden in the big-$O$ notation depends on the choice of $t$. An inspection of the proof shows that the dependence is roughly $\Theta(t^7)$. Firstly, in order to guarantee $1-\binom{t}{2} \cdot \eps$ is not too small, we require $\eps = O(1/t^2)$. The codes promised by \Cref{thm:good-code} can be verified to require $\beta = O(1/\eps^3)$ (this itself stems from using the best-known explicit constructions of lossless expander graphs, whose degree are of this order). Thus, in the end we require $\beta = \Omega(t^6)$, and since the total size is about $t \cdot \beta$, we get $\Theta(t^7)$. In particular, with $t = \poly(\log (k+r))$ we could still achieve size $\tilde O(k+r)$. This means that when $t=\log^{o(1)} k$, we get better asymptotic size than an FFT-based implementation of a hash function using degree-$t$ polynomials.
\end{remark}

\begin{remark}
    We remark that having smaller algebraic degree is not possible, at least in the case $q=2$. Here, we use the following standard fact (see, e.g., \cite[Section~3.1]{HHL19}): if $f:\F_2^\beta \to \F_2$ is a polynomial map with algebraic degree at most $d$ and $d<\beta$, then for any $x_0,x_1,\dots,x_{d+1} \in \F_2^\beta$, we have $\sum_{S \subseteq [d+1]}f\left(x_0+\sum_{i\in S}x_i\right)=0$. That is: the sum over any $(d+1)$-dimensional affine space is $0$.
    
    Now, suppose one had $\cH$ supported on maps $G = (G_1,\dots,G_\beta):\F_2^\beta \to \F_2^\beta$, where each $G_j$ is of degree $d$ with $2^{d+1}\leq t$. Then, choose $x_1,\dots,x_{d+1} \in \F_2^\beta$ linearly independent (this is possible since $d<\beta$, i.e., $d+1\leq \beta$), and $x_0$ arbitrarily; say, $x_0=0$. The linear independence implies the values $\sum_{i\in S}x_i$ are distinct for over $S \subseteq [d+1]$. Then we know that, no matter what $G$ is sampled, we have
    \begin{align*}
        \sum_{S \subseteq [d+1]}G\left(\sum_{i \in S}x_i\right) &= \sum_{S \subseteq [d+1]} \left(G_1\left(\sum_{i \in S}x_i\right),\dots,G_\beta\left(\sum_{i \in S}x_i\right)\right)\\ 
        &= \left(\sum_{S \subseteq [d+1]}G_1\left(\sum_{i \in S}x_i\right),\dots,\sum_{S \subseteq [d+1]}G_\beta\left(\sum_{i \in S}x_i\right)\right) = (0,\dots,0) = 0 \ .
    \end{align*}
    Thus, we can find $2^{d+1}\leq t$ distinct input values such that, after revealing the hash value on the first $2^{d+1}-1$ of them, the last value is determined (as the sum of the previous values). So the distribution cannot define a $t$-wise independent hash function family. 
\end{remark}

\subsection{Fast Bounded-Independence LUOF} \label{sec:t-luof} 
Recall the definitions of a linear uniform output family (LUOF) (\Cref{sec:first_dual}) and a $t$-wise LUOF (\Cref{sec:second_dual}). The $t$-wise LUOF can be thought of as a linear analogue of $t$-wise independent hash functions. Of course, when restricting to a linear function, one cannot hope to achieve a $t$-wise independent hash function, as any distinct yet linearly dependent inputs will be mapped to linearly dependent outputs (and hence not stochastically independent). This motivated us to define a $t$-LUOF as any matrix which maps $t$ distinct \textit{linearly independent} non-zero inputs to uniformly random and independent outputs. 

Our goal in this subsection is to give a construction for a fast $t$-wise LUOF. Recall that we required this for the fast list-decodable codes with fast list-decodable duals of \Cref{sec:second_dual}. Akin to how we showed that the fast $t$-wise independent hash function by generalizing the IKOS construction for fast pairwise independent hash function (assuming one uses a code with large enough distance), we will generalize the fast LUOF construction of Druk and Ishai \cite{DI14} to give a fast $t$-LUOF. Here also this will rely on the codes from \Cref{thm:good-code} having sufficiently large distance. We remark that the encoding circuits for the codes we construct only have linear size if $q = O(1)$ (see \Cref{rmk:eps-t-LUOF}).

\begin{theorem} \label{thm:t-LUOF}
    Let $k,r\in \N$ be growing parameters, fix $t \in \N$, and let $q$ be a prime power. There exists an arithmetic circuit $C_{k,r}:\F_q^k \to \F_q^r$ of size $O(k+r)$ with $O(k+r)$ uniformly random field elements such that, over the $O(k+r)$ uniformly random field elements, the circuit implements a $t$-LUOF.
\end{theorem}

\begin{proof}
     We first provide the construction, which is very similar to that given in the proof of \Cref{thm:t-wise-hash}. Fix $0<\eps,\delta_2<1$ with $1-q^t\cdot \eps+\delta_2>1$. Fix $\beta \in \N$ with $\beta \geq t$ large enough so that the following hold. We first choose codes $\cC_1,\cC_2$ exactly as in \Cref{thm:t-wise-hash}, which we now recall for the reader's convenience. 

    Let $\cC_1 \subseteq \left(\F_q^\beta\right)^m$ and $\cC_2 \subseteq \left(\F_q^\beta\right)^m$ be $\F_q$-additive codes of dimension $k$ and $r$, respectively, and with distance $1-\eps$ and $\delta_2$, respectively. Such codes are promised to exist by \Cref{thm:good-code} (along with \Cref{rem:decreasing-dimension} in order to reduce the codes' dimensions, if required). Let $G_1$ and $G_2$ denote generator matrices for $\cC_1$ and $\cC_2$, respectively. Additionally, from \Cref{thm:good-code} and \Cref{lem:fast_tranpose} it holds that $G_1$ and $G_2^\top$ can both be implemented by uniform arithmetic circuits of size $O(m) = O(k+r)$; call these circuits $C_1$ and $C_2$, respectively.

    Lastly, the circuit samples $m \cdot \beta^2 = O(r+k)$ uniformly random symbols which we take to describe $m$ random matrices $F_1,\dots,F_m \in \F_q^{\beta \times\beta}$. Let $F \in \F_q^{\beta m \times \beta m}$ denote the block-diagonal matrix with the $F_i$'s as blocks, that is,
    \[
        F = \begin{bmatrix} F_1 & 0 & \cdots & 0 \\
        0 & F_2 & \cdots & 0 \\
        \vdots & \vdots & \ddots & \vdots\\
        0 & 0 & \cdots & F_m 
        \end{bmatrix} \ .
    \]
    Observe that the map $\left(\F_q^\beta\right)^{ m} \to \left(\F_q^{\beta}\right)^m$ obtained by mapping $((x^{(1)},\dots,x^{(m)}) \mapsto (x^{(1)}F_1,\dots,x^{(m)}F_m)$ can be implemented by a uniform randomized arithmetic circuit of size $O(m \cdot \beta^2) = O(r+k)$ with the $F_i$'s as above; let $B$ denote this circuit. The final circuit $C_{k,r}$ is then defined as $C_2 \circ B \circ C_1$. That they have size $O(k+r)$ follows from the given construction. The circuit samples $O(r+k)$ uniformly random field elements to implement $F$.

  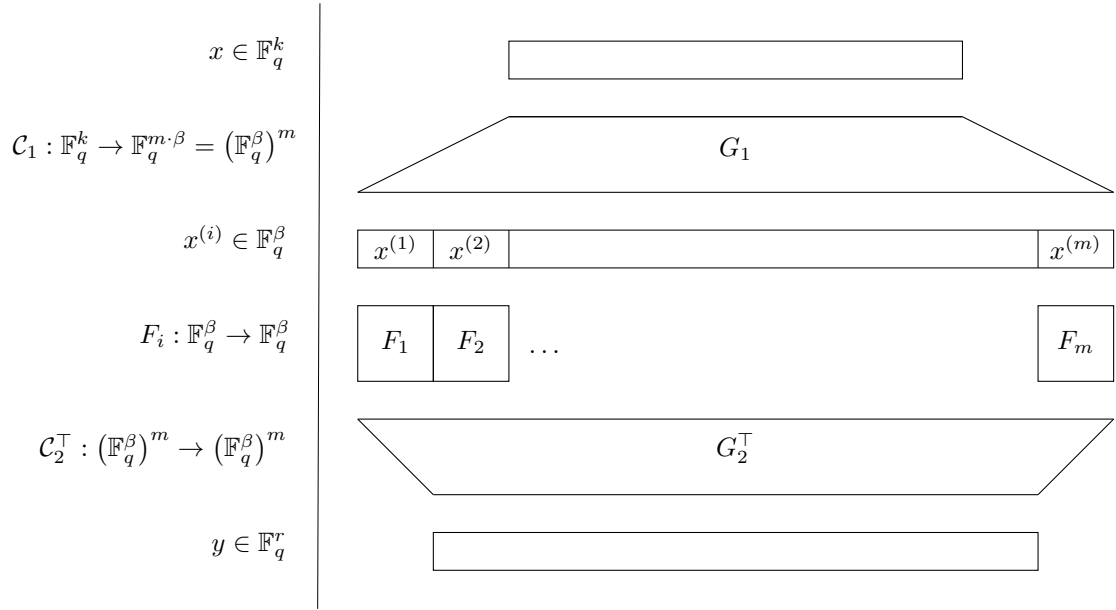
\begin{figure}[!h]
        \centering
        \begin{tikzpicture}[scale=1]
          \coordinate (A) at (0.00, 1.00);
          \coordinate (B) at (1.00, 0.00);
          \coordinate (C) at (9.00, -0.00);
          \coordinate (D) at (10.00, 1.00);
          \coordinate (E) at (0.00, 4.00);
          \coordinate (F) at (2.00, 5.00);
          \coordinate (G) at (8.00, 5.00);
          \coordinate (H) at (10.00, 4.00);
          \coordinate (I) at (5.00, 4.30);
          \coordinate (J) at (5.00, 0.30);
          \coordinate (K) at (0.50, 1.75);
          \coordinate (L) at (1.50, 1.75);
          \coordinate (M) at (2.50, 1.75);
          \coordinate (N) at (9.50, 1.75);
          \coordinate (O) at (-1.46, 5.50);
          \coordinate (P) at (-2.68, 4.25);
          \coordinate (W) at (9.50, 3.00);
          \coordinate (X) at (1.50, 3.00);
          \coordinate (Y) at (0.50, 3.00);
          \coordinate (Z) at (-1.62, 3.00);
          \coordinate (AA_1) at (-2.57, 0.25);
          \coordinate (AA) at (-1.89, 1.75);
          \coordinate (AB) at (-1.44, -1.00);
          \coordinate (Q) at (1.00, 3.50);
          \coordinate (R) at (1.00, 3.00);
          \coordinate (S) at (2.00, 3.50);
          \coordinate (T) at (2.00, 3.00);
          \coordinate (U) at (9.00, 3.50);
          \coordinate (V) at (9.00, 3.00);
          \coordinate (AC) at (-0.50, 6.50);
          \coordinate (AD) at (-0.53, -1.51);
        
          \draw (A) -- (B);
          \draw (B) -- (C);
          \draw (D) -- (C);
          \draw (A) -- (D);
          \draw (0,1.5) rectangle (1,2.5);
          \draw (1,1.5) rectangle (2,2.5);
          \draw (9,1.5) rectangle (10,2.5);
          \draw (E) -- (F);
          \draw (F) -- (G);
          \draw (F) -- (G);
          \draw (G) -- (H);
          \draw (E) -- (H);
          \draw (2,5.5) rectangle (8,6);
          \node[above] at (I) {$G_1$};
          \node[above] at (J) {$G_2^\top$};
          \node[above] at (K) {$F_1$};
          \node[above] at (L) {$F_2$};
          \node[above] at (M) {$\dots$};
          \node[above] at (N) {$F_m$};
          \node[above] at (O) {$x \in \mathbb{F}_q^k$};
          \draw (0,3) rectangle (10,3.5);
          \node[above] at (P) {$\cC_1 : \mathbb{F}_q^k \to \mathbb{F}_q^{m\cdot \beta} = \left(\mathbb{F}_q^\beta\right)^m$};
          \node[above] at (W) {$x^{(m)}$};
          \node[above] at (X) {$x^{(2)}$};
          \node[above] at (Y) {$x^{(1)}$};
          \node[above] at (Z) {$x^{(i)}\in \mathbb{F}_q^\beta$};
          \node[above] at (AA_1) {$\cC_2^\top : \left(\mathbb{F}_q^\beta\right)^m \to \left(\mathbb{F}_q^\beta\right)^m$};
          \node[above] at (AA) {$ F_i : \mathbb{F}_q^\beta \to \mathbb{F}_q^\beta$};
          \draw (1,-1) rectangle (9,-0.5);
          \node[above] at (AB) {$y \in \mathbb{F}_q^r$};
          \draw (Q) -- (R);
          \draw (S) -- (T);
          \draw (U) -- (V);
          \draw[ultra thin] (AC) -- (AD);
        \end{tikzpicture}
        \caption{A schematic of the IKOS construction for $t$-linear uniform output families ($t$-LUOF). The functions $F_i$ are constant size random matrices which we sample. $\cC_1$ is a code with generator matrix $G_1$ and $\cC_2$ is a code with generator matrix $G_2$.}
        \label{fig:ikos_construction_2}
    \end{figure}

    Let us write $A = G_1 F G_2^\top$ to denote the function defined by the circuit $C_{k,r}$. We refer the reader to \Cref{fig:ikos_construction_2} which gives a schematic of $A$. What we want to prove is that, over the random field elements sampled by the circuit to implement $F$, $A$ is a $t$-LUOF. Let $x_1,\dots,x_t \in \F_q^k$ be linearly independent. By \Cref{lem:inner-products-uniform}, it suffices to consider a nonzero $y \in \F_q^{tr}$, and prove that $\langle y,(Ax_1,\dots,Ax_t)\rangle \sim \F_q$ over the $O(k+r)$ uniformly random field elements. Write $y = (y_1,\dots,y_t)$ with each $y_i \in \F_q^r$. Then
    \[
        \langle y,(x_1A,\dots,x_tA)=\sum_{i=1}^t \left\langle y_i, x_iG_1 F G_2^\top\right\rangle = \sum_{i=1}^t \left\langle y_iG_2, x_iG_1 F\right\rangle\ .
    \]
    Put $d_i := y_iG_2$ and $c_i := x_i G_1$ for all $i \in [t]$. Note that the vectors $c_1,\dots,c_t \in \cC_1$, and that they are all linearly independent. We additionally have that $d_1,\dots,d_t \in \cC_2$, and they are not all $0$. Without loss of generality, assume that $d_1 \neq 0$. For each $i \in [t]$, write $d_i = (d_i^{(1)},\dots,d_i^{(m)})$ and $c_i = (c_i^{(1)},\dots,c_i^{(m)})$ with each $d_i^{(\ell)},c_i^{(\ell)}\in\F_q^\beta$. Then
    \[
        \sum_{i=1}^t \left\langle y_iG_2, x_iG_1 F\right\rangle = \sum_{\ell=1}^m \sum_{i=1}^t \left\langle d_i^{(\ell)}, c_i^{(\ell)} F_\ell\right\rangle
    \]
    We claim there exists an $\ell^* \in [m]$ for which:
    \begin{itemize}
        \item[(a)] $d_1^{(\ell^*)} \neq 0$, and
        \item[(b)] the vectors $c_1^{(\ell^*)},\dots,c_t^{(\ell^*)} \in \F_q^\beta$ are \emph{linearly independent}.
    \end{itemize}
    To see this, since $\cC_2$ has distance $\geq \delta_2$, condition (a) holds for at least a $\delta_2$ fraction of $\ell \in [m]$. Considering now condition (b), for each $\lambda \in \F_q^t\setminus \{0\}$, we have that $c_\lambda = \sum_{i=1}^t \lambda_i c_i$ is a nonzero codeword (that they are nonzero uses the assumption that the $c_i$'s are linearly independent). Thus, for at least a $1-\eps$ fraction of $\ell \in [m]$, we have $c_\lambda^{(\ell)} \neq 0$, or in other words, for at most an $\eps$ fraction of $\ell \in [m]$ we have $c_\lambda^{(\ell)} = 0$. By a union bound, the fraction of $\ell \in [m]$ for which $c_\lambda^{(\ell)}=0$ for \emph{some} $\lambda \in \F_q^t\setminus\{0\}$ is at most $q^t \cdot \eps$. Hence, condition (b) holds for at least a $1-q^t \cdot \eps$ fraction of $\ell \in [m]$. Since we assumed $1-q^t \cdot \eps + \delta_2 > 1$, we find that the required $\ell^* \in [m]$ exists.

    Fix one such $\ell^* \in [m]$. Since the $F_1,\dots,F_m$ are independent, if we prove that $\sum_{i=1}^t \left\langle d_i^{(\ell^*)}, c_i^{(\ell^*)} F_{\ell^*}\right\rangle \sim \F_q$, then also $\sum_{\ell=1}^m \sum_{i=1}^t \left\langle d_i^{(\ell)}, c_i^{(\ell)} F_\ell\right\rangle \sim \F_q$, as desired. To prove $\sum_{i=1}^t \left\langle d_i^{(\ell^*)}, c_i^{(\ell^*)} F_{\ell^*}\right\rangle \sim \F_q$, suppose without loss of generality that we have $d_1^{(\ell^*)},\dots,d_s^{(\ell^*)}\neq0$ and $d^{(\ell^*)}_{s+1}=\cdots=d_t^{(\ell^*)}=0$. Then
    \[
        \sum_{i=1}^t \left\langle d_i^{(\ell^*)}, c_i^{(\ell^*)} F_{\ell^*}\right\rangle = \sum_{i=1}^s \left\langle d_i^{(\ell^*)}, c_i^{(\ell^*)} F_{\ell^*}\right\rangle \ .
    \]
    Now, since the $c_1^{(\ell^*)},\dots,c_s^{(\ell^*)}$ are linearly independent, over the randomness of $F_{\ell^*}$, the vectors $c_1^{(\ell^*)}F_{\ell^*},\dots,c_s^{(\ell^*)}F_{\ell^*}$ are independent and uniform over $\F_q^\beta$. Therefore, applying \Cref{lem:inner-products-uniform}, the inner products $\left\langle d_i^{(\ell^*)}, c_i^{(\ell^*)} F_{\ell^*}\right\rangle$ are independent and uniform for $i=1,\dots,s$, so the sum $\sum_{i=1}^s \left\langle d_i^{(\ell^*)}, c_i^{(\ell^*)} F_{\ell^*}\right\rangle$ is uniform over $\F_q$, as desired. 
\end{proof}

\begin{remark}\label{rmk:eps-t-LUOF}
	Again, as in \Cref{rmk:eps-t-hash} we can roughly say that since we require $\eps < q^{-t}$ and $\beta > 1/\eps^3$, we at the very least require $\beta > q^{3t}$, showing that the constant in the big-$O$ notation for the circuits in \Cref{thm:t-LUOF} is at least $q^{O(t)}$. In particular, this forces both $q$ and $t$ to be constant to obtain a meaningful bound. 
\end{remark}

\section{Cryptographic Applications}
\label{sec:apps}

In this section we discuss several cryptographic applications of fast bounded-independence functions, most of which also require fast duals.

\subsection{Fast Information-Theoretic MPC}

We start by applying asymptotically good fast codes with fast duals towards secure multiparty computation (MPC) in the client-server model. For constant-size functions, our protocol is perfectly secure against almost 1/2 of the $n$ servers, and has total circuit complexity of $O(n)$. This should be contrasted with similar protocols based on algebraic MPC-friendly codes, such as Reed-Solomon codes, whose circuit complexity is (at least) quasi-linear in $n$.

\begin{theorem}[Near-optimal perfect MPC for constant-size functions] \label{thm:mpc}
Let $m$ be a constant number of clients, $X$ and $Y$ finite input and output domains, and $\eps>0$ an arbitrarily small constant. For every (finite) $f:X^m \to Y^m$ and $\eps>0$, there is an $m$-client, $n$-server protocol for $f$ with the following features.
\begin{itemize}
    \item The protocol has perfect (information-theoretic) security against a passive adversary corrupting (at most) one client and $t=(1/2-\eps)\cdot n$ of the $n$ servers. 
    \item The total computational complexity of all parties, measured by Boolean circuit size, is $O(n)$. 
\end{itemize}
\end{theorem}

The above result is essentially the best one could hope for in the case of constant-size functions.
First, the security threshold is nearly optimal, since a similar protocol with security threshold $t>n/2$ would imply information-theoretic oblivious transfer. Second, the circuit complexity is asymptotically optimal, since if a client sends messages to $o(n)$ servers, the adversary can corrupt all of these servers and learn the client's input. (In the perfect security case, this holds even in random access models where we only charge for messages actually sent.) 

Protocols in the above client-server model can be used as a building block for other cryptographic applications. For example, the above result can be used to obtain an {\em oblivious transfer combiner} with similar near-optimality features using the MPC-based approach from~\cite{HIKN08}.

\paragraph{An OLE protocol.} We will start with the simpler special case of (1) $m=2$ clients, (2) the function $f={\sf OLE}$ that takes one bit $x$ from a receiver client, two bits $(y,z)$ from a sender client and delivers $xy+z$ (over $\F$) to the receiver, and (3) sub-optimal but constant fractional security threshold $t=\Omega(n)$. This case already captures the core difficulty of the problem. We will later outline how to bootstrap from this case to general functions and near-optimal security threshold using standard techniques. In the following $\F$ can be any finite field, but the case $\F=\F_2$ is sufficient for our purposes.

\begin{proposition}
\label{prop:ole}
Let $f: \F \times \F^2 \to \F$ be the OLE functionality, which takes $x\in\F$ from a Receiver client and $y,z\in\F$ from a Sender client, and delivers $f(x, (y,z)) = xy+z$ to Receiver. Let $\mathcal{C} \le \mathbb{F}^{n+1}$ and its dual $\mathcal{C}^\perp \le \mathbb{F}^{n+1}$ be linear codes with coordinates indexed $0, 1, \dots, n$. Assume the following:
\begin{enumerate}
    \item $\mathcal{C}$ has minimum distance $\dist> t + 1$.
    \item $\mathcal{C}^\perp$ has minimum distance $\dist^\perp> t + 1$.
    \item Both $\mathcal{C}$ and $\mathcal{C}^\perp$ have encoder circuits, $E$ and $E^\perp$, of size $s$.
\end{enumerate}
Then, there exists a perfectly secure protocol that computes $f$ against any passive adversary corrupting up to $t$ servers and at most one client. Furthermore, the total arithmetic circuit size of all parties is $O(n+s)$.
\end{proposition}

To prove the above proposition, we start with a lemma that applies the well-known connection between linear error correcting codes and secret sharing~\cite{massey1995some} while respecting the circuit complexity of the encoder. See Appendix~\ref{app:sharing} for a proof.

\begin{lemma}[Secret sharing from linear codes]
\label{lem:sharing}
Let $\mathcal{C} \le \mathbb{F}^{n+1}$ be a linear code with coordinates indexed $0, 1, \dots, n$ and dual distance $\dist^{\perp} \ge 2$. Let $E: \mathbb{F}^k \to \mathbb{F}^{n+1}$ be an arithmetic encoder circuit for $\mathcal{C}$ of size $s$. Then, the secret sharing scheme that shares a secret $\sigma \in \mathbb{F}$ into $n$ shares by sampling a uniformly random codeword $c \in \mathcal{C}$ subject to $c_0 = \sigma$ and dropping $c_0$ satisfies the following:
\begin{enumerate}
    \item Perfect privacy: Any subset of shares $T \subset \{1, \dots, n\}$ of size $|T| \le \dist^{\perp} - 2$ perfectly hides $\sigma$.
    \item Linear complexity: The generation of the shares from $\sigma$ can be computed by a randomized arithmetic circuit of size at most $s + n$.
\end{enumerate}
\end{lemma}

We are now ready to describe the protocol establishing~\Cref{prop:ole}, using a simple client-server instance of MPC  from ``multiplicative'' secret sharing~\cite{CDM00}.

\begin{protocol}[Information-Theoretic OLE over $\mathbb{F}$]\label{protocol}
\ 
\begin{enumerate}
    \item \textbf{Setup:} Client 1 (Receiver) holds input $x \in \mathbb{F}$. Client 2 (Sender) holds inputs $y, z \in \mathbb{F}$. There are $n$ servers. A passive adversary may corrupt at most one client and $t$ servers.

    \item \textbf{Receiver input sharing:}
    \begin{enumerate}
        \item Receiver samples a uniformly random codeword $c \in \mathcal{C}$ subject to the constraint $c_0 = x$. This is done using~\Cref{lem:sharing}.
        \item  Receiver sends share $c_i$ to Server $i$ (for $i = 1, \dots, n$).
    \end{enumerate}
    
    \item \textbf{Sender input sharing:}
    \begin{enumerate}
        \item Sender samples a uniformly random codeword $d \in \mathcal{C}^\perp$ subject to the constraint $d_0 = y$. This too is done using~\Cref{lem:sharing}.
        \item Sender samples a uniformly random blinding vector $e = (e_1, \dots, e_n) \in \mathbb{F}^n$ subject to the strict additive constraint $\sum_{i=1}^n e_i = z$.
        \item Sender sends shares $d_i$ and $e_i$ to Server $i$ (for $i = 1, \dots, n$).
    \end{enumerate}
    
    \item \textbf{Server computation:}
    \begin{enumerate}
        \item Server $i$ receives $c_i$ from Receiver and $(d_i, e_i)$ from Sender.
        \item Server $i$ locally computes $w_i = e_i - c_i \cdot d_i$.
        \item Server $i$ sends $w_i$ back to  Receiver.
    \end{enumerate}
    
    \item \textbf{Output decoding:}
    \begin{enumerate}
        \item Receiver computes the final output as $\text{out} = \sum_{i=1}^n w_i$.
    \end{enumerate}
\end{enumerate}
\end{protocol}

See Appendix~\ref{app:protocol} for analysis of~\Cref{protocol}.

\paragraph{From OLE to general constant-size functions.} To extend the above protocol from two-party OLE to a  protocol with $m=O(1)$ clients for any constant-size function $f$, we can use the known completeness of OLE (an arithmetic variant of OT) for secure computation of Boolean or arithmetic circuits~\cite{GMW,IPS09}. At a high level, the clients start by additively sharing their inputs among each other. Then, each addition gate is processed by having the clients locally add their shares, and each multiplication gate is processed via $O(m^2)$ invocations of the OLE protocol to additively share each cross-term $a_ib_j$ between the client holding $a_i$ and the client holding $b_j$. The protocol concludes by having the clients exchange their output shares. When the field $\F$ is of constant size, the Boolean circuit complexity of the protocol is $O(n)$ for any constant-size circuit.

\paragraph{Achieving near-optimal security.} The security threshold $t$ of the protocol is close to the smaller minimal distance of the primal code and its dual, namely $\min(\dist,\dist^\perp)$. One approach for getting it arbitrarily close to $1/2$ is to let $\F$ be an arbitrarily large constant, in which case \Cref{thm:first_dual} gives a pair of codes with relative distance $1/2-\epsilon$. Another approach is to start with a protocol that has an arbitrarily small (but constant) security threshold, say by using $\F=\F_2$, and apply a general threshold amplification technique based on virtualization~\cite{bracha1987randomized,OstrovskyRV94,damgard2008scalable}.

\paragraph{Putting everything together.} With the above extensions of~\Cref{prop:ole}, we can use our constructions of fast dual codes (cf.~\Cref{thm:first_dual}) to obtain the main result of this section captured by~\Cref{thm:mpc}. This construction is semi-explicit in the sense that there is a PPT algorithm that generates circuits implementing the parties with $negl(n)$ failure probability.

\subsection{Fast Encrypted Matrix-Vector Product}

Switching to the computational security setting, a recent application of fast dual codes with random-like properties was recently given in the context of computing on encrypted data~\cite{BenhamoudaCHIKM25}. Concretely, an {\em encrypted matrix-vector product} (EMVP) protocol enables a client to encode a matrix $M$ using a secret code $\cal C$ and then compute an unbounded number of matrix-vector products $Mq_i$ by encrypting each query vector $q_i$ using the dual code ${\cal C}^\perp$. Linear algebra attacks are prevented by adding a suitable type of noise. Under a variant of the Learning Subspace with Noise (LSN) assumption~\cite{DodisKL09,CIMR}, the protocol hides both $M$ and $q_i$ from a server who stores the encrypted matrix and processes the encrypted queries.

Using our new families of fast dual codes, we obtain the first candidate EMVP protocols in which both the matrix encoding and the query encoding have asymptotically optimal circuit complexity. The recent construction of fast dual codes from~\cite{BR26} is limited to $\F_2$ and has a non-negligible failure probability.

\subsection{Fast and Conservative Substitutes for Random Linear Codes}
\label{sec:substitute}

The previous two applications inherently relied on having a fast code coupled with a fast dual. A broader use case, which does not require duality, is providing a fast substitute to a random linear code in the context of applications that rely on the hardness of ``noisy linear algebra.'' Our new construction of fast $t$-LUOF (\Cref{sec:t-luof}) can serve as a more conservative candidate than the previous 1-LUOF construction from~\cite{DI14}.

In more detail, in most cryptographic applications of noisy linear algebra, one typically considers the minimal distance of the code or its dual as a heuristic measure of hardness. Thus, it is perhaps not clear why the $t$-local similarity provided by $t$-LUOF (for $t\geq 2$) is interesting. To illustrate this, consider code-based constructions that rely on the LPN assumption,\footnote{What we will give below is often called \emph{dual}-LPN, which is equivalent to the more standard variant.} where one assumes that the distributions
\[
    (A,Ae) ~\text{ and }~ (A,b)
\]
are computationally indistinguishable, where $e$ is a random sparse vector and $b$ is a uniformly random vector (in this case, one should think of $A \in \F_q^{k \times r}$ with $r>k$). This use of minimal distance as a proxy for conjectured LPN hardness is common in the recent line of works on pseudorandom correlation generators (PCGs)~\cite{BoyleCGI18,boyle2019efficient}. Indeed, minimal distance is sufficient for provably defeating any attack that fits under a natural ``linear attack'' framework~\cite{CouteauRR21}.

However, it is arguably a more conservative to use an $A$ that shares more nuanced combinatorial properties with random linear codes, such as list-decodability/-recoverability. This is exactly what the $t$-LUOF property guarantees. 

\subsection{MPC-Friendly $t$-Wise Hash Functions}

Finally, when distributing the computation of a function $f$ between two or more parties, either via classical MPC techniques~\cite{Yao,GMW,BGW,CCD} or via fully homomorphic encryption~\cite{Gentry09}, the cost of the protocol depends on both the circuit size and the algebraic degree. For $t=O(1)$, our constructions of fast $t$-wise independent hash functions from~\Cref{sec:t-wise-hash} {\em simultaneously} optimize both asymptotic size and degree. 

There are many algorithmic applications that benefit from $t$-wise independent hashing for constant $t>2$. Two well-known examples are the use of 4-wise independence in sketching~\cite{alon1999space} and 5-wise independence in linear probing~\cite{pagh2007linear}.

\bibliography{refs}

\appendix

\section{Deferred Proofs}

\subsection{Proof of~\Cref{lem:sharing}}
\label{app:sharing}

\begin{proof}
We prove each property separately.

\paragraph{Perfect privacy.}
Suppose for the sake of contradiction that a set of shares $T$ with $|T| \le \dist^{\perp} - 2$ leaks information about the secret $\sigma$. In a linear secret sharing scheme, this implies that the secret can be reconstructed via a linear combination of the shares in $T$. Thus, there exist constants $\{a_i\}_{i \in T}$ over $\mathbb{F}$ such that for all valid codewords $c \in \mathcal{C}$:
\[ c_0 = \sum_{i \in T} a_i c_i \]
Rearranging, we get the following linear constraint:
\[ c_0 - \sum_{i \in T} a_i c_i = 0 \]
By definition, any such a constraint corresponds to a vector in the dual code $\mathcal{C}^{\perp}$. Specifically, it defines a dual codeword $h \in \mathcal{C}^{\perp}$ where $h_0 = 1$, $h_i = -a_i$ for $i \in T$, and $h_j = 0$ for all $j \notin T \cup \{0\}$.

Because $h_0 \neq 0$, $h$ is a non-zero dual codeword. Its Hamming weight is at most $1 + |T|$. Given our assumption that $|T| \le \dist^{\perp} - 2$, the weight of $h$ is bounded by:
\[ wt(h) \le 1 + |T| \le 1 + (\dist^{\perp} - 2) = \dist^{\perp} - 1 \]
This directly contradicts the premise that the dual distance of $\mathcal{C}$ is $\dist^{\perp}$, which requires all non-zero dual codewords to have a weight of at least $\dist^{\perp}$. Therefore, no such linear combination exists, and the shares in $T$ are perfectly independent of $\sigma$.

\paragraph{Circuit size.} Using column vector notation, encoder circuit $E$ implements a linear map $c = M m$, where $m \in \mathbb{F}^k$ is a message vector and $M \in \mathbb{F}^{(n+1) \times k}$ is the generator matrix of $\mathcal{C}$. The first coordinate of the resulting codeword is:
\[ c_0 = \sum_{j=1}^k M_{0,j} m_j \]
Because $\dist^{\perp} \ge 2$, the vector $(1, 0, \dots, 0)$ is strictly excluded from $\mathcal{C}^{\perp}$. Consequently, $c_0$ is not identically zero for all codewords, meaning there is at least one index $j^*$ where the matrix entry $M_{0,j^*} \neq 0$. Without loss of generality, assume $M_{0,k} \neq 0$.

To enforce the condition $c_0 = \sigma$ and uniformly sample the remaining degrees of freedom, the arithmetic sharing circuit performs the following steps:
\begin{enumerate}
    \item Takes $\sigma$ and $k-1$ uniformly random field elements $m_1, \dots, m_{k-1} \in \mathbb{F}$ as input.
    \item Computes the dependent message symbol $m_k$ algebraically:
    \[ m_k = M_{0,k}^{-1} \left( \sigma - \sum_{j=1}^{k-1} M_{0,j} m_j \right) \]
    Because the coefficients $M_{0,j}$ are fixed constants, computing $m_k$ requires at most $k-1$ field additions/subtractions and $1$ scalar multiplication. This step adds exactly $k$ arithmetic gates.
    \item Feeds $m = (m_1, \dots, m_k)$ into the original encoder circuit $E$, which requires $s$ gates.
\end{enumerate}
The total size of the new sharing circuit is $s + k$. Since $\mathcal{C}$ is a proper subspace of $\mathbb{F}^{n+1}$ (implied by the dual distance strictly greater than 1), its dimension $k$ is bounded by $k \le n$. Therefore, the total arithmetic circuit size is strictly bounded by $s + n$.
\end{proof}

\subsection{Analysis of~\Cref{protocol}}
\label{app:protocol}

We separately analyze correctness, circuit complexity, and security.

\paragraph{Perfect correctness.}
Receiver computes $\text{out} = \sum_{i=1}^n w_i = \sum_{i=1}^n (e_i - c_i d_i)$.
By linearity, this separates into:
\[ \text{out} = \sum_{i=1}^n e_i - \sum_{i=1}^n c_i d_i \]
Because $c \in \mathcal{C}$ and $d \in \mathcal{C}^\perp$, their inner product over all $n+1$ coordinates is exactly zero: $\sum_{i=0}^n c_i d_i = 0$. This implies $\sum_{i=1}^n c_i d_i = -c_0 d_0 = -xy$.
Furthermore, Sender constructed the blinding vector such that $\sum_{i=1}^n e_i = z$. Substituting these values yields:
\[ \text{out} = z - (-xy) = xy + z \]
as required.

\paragraph{Circuit complexity.} Using~\Cref{lem:sharing}, sampling $c \in \mathcal{C}$ given $c_0=x$ using the size-$s$ encoder $E$ requires an arithmetic circuit of size at most $s+n$. Symmetrically, sampling $d \in \mathcal{C}^\perp$ requires size at most $s+n$. Sampling the additive shares $e$ requires generating $n-1$ random elements and $n-1$ subtractions, which takes $O(n)$ gates. Each of the $n$ servers performs exactly one multiplication and one subtraction ($O(1)$ per server, $O(n)$ total). The Receiver performs $n-1$ additions to reconstruct the output ($O(n)$ gates). Summing these costs yields a total arithmetic circuit size of $O(n+s)$.

\paragraph{Perfect security.} We show that the adversary's view can be perfectly simulated. Let $T \subset \{1, \dots, n\}$ be the set of corrupted servers, with $|T| \le t$. If only the servers are corrupted, security follows immediately from~\Cref{lem:sharing}. We consider the cases where one client is corrupted, and assume wlog that $|T|=t$.

\vspace{0.5em}
\noindent\textit{Case 1: Sender and $t$ servers are corrupted.} The real view consists of $y, z$, the codeword $d$, the blinding vector $e$, and the shares $\{c_i\}_{i \in T}$. The privacy of the Receiver's input depends on the dual distance of $\mathcal{C}$, which is $\dist^\perp$. Because $\dist^\perp > t + 1$, we have $|T| \le t \le d^\perp - 2$. By~\Cref{lem:sharing}, the projection of $c$ onto the coordinates perfectly hides $x$.

\vspace{0.5em}
\noindent\textit{Case 2: Receiver and $t$ servers are corrupted.}
The real view consists of $x$, the codeword $c$, the output $xy+z$, the shares $\{d_i, e_i\}_{i \in T}$ from  Sender, and the responses $\{w_i\}_{i=1}^n$. The privacy of Sender's input $y$ depends on the dual distance of $\mathcal{C}^\perp$, which is the minimum distance of the primal code $\mathcal{C}$, denoted as $\dist$. 
\begin{enumerate}
    \item \textbf{Simulate Sender's shares:} Because $\dist > t + 1$, we have $|T| \le t \le \dist - 2$. The true projection of $d$ onto $T$ is uniformly distributed over $\mathbb{F}^{|T|}$. The simulator samples $\{d_i\}_{i \in T}$ uniformly at random. It also samples $\{e_i\}_{i \in T}$ uniformly at random.
    \item \textbf{Simulate corrupted server responses:} For $i \in T$, the simulator computes $w_i = e_i - c_i d_i$.
    \item \textbf{Simulate honest server responses:} The final sum must evaluate to the legitimate output $xy+z$. The sum of the uncorrupted responses is deterministically constrained: $\sum_{i \notin T} w_i = (xy+z) - \sum_{i \in T} w_i$. The simulator samples the remaining $n - |T|$ values of $w_i$ uniformly at random subject to this sum.
\end{enumerate}
In the real protocol, the uncorrupted blinding shares $\{e_i\}_{i \notin T}$ are uniformly random field elements subject only to the constraint $\sum_{i \notin T} e_i = z - \sum_{i \in T} e_i$. Because $w_i = e_i - c_i d_i$, the $\{e_i\}_{i \notin T}$ values act as an information-theoretic one-time pad perfectly masking the values of $c_i d_i$. The constraint on the visible responses is:
\[ \sum_{i \notin T} w_i = \sum_{i \notin T} (e_i - c_i d_i) = \left( z - \sum_{i \in T} e_i \right) - \left( -xy - \sum_{i \in T} c_i d_i \right) = (xy+z) - \sum_{i \in T} (e_i - c_i d_i) = (xy+z) - \sum_{i \in T} w_i \]
The real variables $\{w_i\}_{i \notin T}$ are uniformly distributed subject exactly to this equation, perfectly matching the simulator. The adversary extracts no additional information about $y$ or $z$.

\end{document}